\documentclass[a4paper, 12 pt, twoside, reqno]{amsart}

\usepackage[english]{babel}
\usepackage[utf8]{inputenc}
\usepackage{slashed}
\usepackage[all]{xy}
\usepackage{setspace}
\usepackage[euler-digits,euler-hat-accent]{eulervm}
\usepackage{times}
\usepackage{amsmath}
\usepackage{amsfonts}
\usepackage{amssymb}
\usepackage{amsmath}
\usepackage{amsthm}
\usepackage{graphicx}
\usepackage{array}
\usepackage{color}
\usepackage{mathrsfs}
\usepackage{hyperref}
\usepackage{eucal}
\usepackage{esint} %%% for \fint
\usepackage{tikz}
\usepackage{upgreek}
\usepackage{enumitem}

\allowdisplaybreaks

\theoremstyle{plain}
\newtheorem{theorem}{Theorem}[section]
\newtheorem{corollary}[theorem]{Corollary}
\newtheorem{proposition}[theorem]{Proposition}
\newtheorem{lemma}[theorem]{Lemma}

\theoremstyle{definition}
\newtheorem{definition}[theorem]{Definition}
\newtheorem{assumption}[theorem]{Assumption}

\theoremstyle{remark}
\newtheorem{remark}[theorem]{Remark} 
\newtheorem{example}[theorem]{Example}

\numberwithin{equation}{section}
\numberwithin{figure}{section}
\numberwithin{table}{section}

\newcommand{\R}{\mathbb{R}}
\newcommand{\N}{\mathbb{N}}
\newcommand{\C}{\mathbb{C}}                           

\newcommand{\Z}{\mathbb{Z}}
\newcommand{\T}{\mathbb{T}}

\newcommand{\s}[1]{\CMcal{#1}}
                  
\newcommand{\bb}[1]{\mathscr{#1}}
\newcommand{\rr}[1]{\mathfrak{#1}}
\newcommand{\n}[1]{\mathbb{#1}}

\newcommand{\expo}[1]{\,\mathrm{e}^{#1}\,}                 

\newcommand{\dd}{\,\mathrm{d}}
\newcommand{ \ii}{\,\mathrm{i}\,}

\newcommand{\virg}[1]{\lq\lq#1\rq\rq}                \newcommand{\ie}{\textsl{i.\,e.\,}}
\newcommand{\eg}{\textsl{e.\,g.\,}}
\newcommand{\cf}{\textsl{cf}.\,}

\newcommand{\A}{\mathfrak{A}}

\hypersetup{
pdftoolbar=true,        % show Acrobat’s toolbar?
pdfmenubar=true,        % show Acrobat’s menu?
pdffitwindow=true,     % window fit to page when opened
pdfstartview=true,    % fits the width of the page to the window
pdftitle={NCG-Landau-II},    % title
pdfauthor={Giuseppe De Nittis},     % author
breaklinks=true, %permet le retour a la ligne dans les liens trop longs
colorlinks=true,       % false: boxed links; true: colored links
linkcolor=purple,         % color of internal links (change box color 
citecolor=teal, % color of links to bibliography
urlcolor=blue, %couleur des hyperliens
bookmarksopen=true, %si les signets Acrobat sont crees,
filecolor=magenta,      % color of file links
}

\begin{document}

\title[About the notion of  eigenstates for $C^*$-algebras and some applications]{About the notion of  eigenstates for $C^*$-algebras and some application in Quantum Mechanics}

\author[G. De~Nittis]{Giuseppe De Nittis}

\address[G. De~Nittis]{Facultad de Matem\'aticas \& Instituto de F\'{\i}sica,
  Pontificia Universidad Cat\'olica de Chile,
  Santiago, Chile.}
\email{gidenittis@mat.uc.cl}

\author[D. Polo]{Danilo polo}

\address[D. Polo]{Facultad de Matem\'aticas,
  Pontificia Universidad Cat\'olica de Chile,
  Santiago, Chile.}
\email{djpolo@mat.uc.cl}

\vspace{2mm}

\date{\today}

\begin{abstract}
This work is concerned with the notion of \emph{eigenstates} for $C^*$-algebras.
After reviewing some basic and structural results, 
we
explore the possibility of reinterpreting certain typical concepts of quantum mechanics (\eg dynamical equilibrium states, ground states, gapped states,  Fermi surfaces) in terms of (algebraic) eigenstates. 
\medskip


\end{abstract}

\maketitle

\tableofcontents

%--------------------%
\section{Introduction}\label{sect:intro}
Usually, the concept of \emph{eigenstate}  is formulated within the Hilbert space formulation of  Quantum Mechanics (QM) \cite{dirac-30,vonneumann-32}: given a linear operator $H$ (commonly called \emph{observable}) acting on the Hilbert space $\s{H}$ (the \emph{state} space) and a number  $\lambda\in\C$ (the \emph{measure outcome}), then one says that $\psi\in \s{H}$ is an \emph{eigenstate} of $H$ with \emph{eigenvalue} $\lambda$ if the equation $H\psi=\lambda \psi$ is satisfied.
However, immediately after the birth of the Hilbert space formulation of QM, it was realized that the quantum theory can be rigorously reformulated 
in a purely algebraic language. The \emph{algebraic foundation of QM} \cite{murray-von-neumann-36,segal-47} turns out to be  a more general and flexible theory compared with  the original Hilbert space formulation.
In this setting the observables are interpreted as elements of a $C^*$-algebra $\rr{A}$ and the states are the continuous linear funcional $\omega:\rr{A}\to\C$. The connection with the 
standard formulation is given by the \emph{GNS representation} (see Section \ref{sec:eig-rep})
which provides a way to associate to a given state $\omega$ a triple $(\pi_\omega, \mathcal{H}_\omega, \psi_\omega)$ where $\pi_\omega$ is a representation of $\rr{A}$ over the Hilbert space $\mathcal{H}_\omega$, and $\psi_\omega$ is a cyclic vector $\psi_\omega$ which  permits to reconstruct the statistic (or \emph{expectation values}) of $\omega$. One of the main advantages of the algebraic approach is that it provides a framework suitable for dealing without ambiguity with systems with infinite degrees of freedom such as the statistical mechanics systems \cite{bratteli-robinson-87,bratteli-robinson-97}   and field theories \cite{haag-96}.

\medskip

 In light of the above, it is worth asking whether there is an algebraic interpretation of the concept of eigenstate of an operator. In fact this is not so exotic or difficult. Given an observable $H$ of a $C^*$-algebra $\rr{A}$ and a number $\lambda\in\C$ one says that 
a state $\omega$ is an \emph{eigenstate} of $H$ if  the eigenvalue equation $\pi_\omega(H)\psi_\omega=\lambda \psi_\omega$ 
is satisfied in the related GNS representation 
(Theorem \ref{teo_rep}). 
Remarkably,
this somehow natural definition of an eigenstate can be characterized in a purely algebraic way without resorting to the  GNS representation (Definition \ref{def_00}). It is not clear to us where and when this definition was first introduced in the literature. However, it has been used by various authors for distinct reasons. For instance in \cite{riedel-88,riedel-90} it is used for the study of spectral properties of the \emph{almost Mathieu operator} and in \cite{Pash1, Pash2} for the study of the simplicity of certain $C^*$-algebras. A recent review \cite{rinehart-21} is devoted to summarize the main properties of the eigenstates for  self-adjoint operators.
\medskip

 \medskip

This work is aimed to review and popularize the concept of  eigenstates for   elements of a $C^*$-algebra, and to provide a possible use of this notion to reinterpret certain ideas typical for (quantum)  condensed matter systems.
 By following, and extending, the presentation in \cite{rinehart-21}, we will  explore the relation between the notion of eigenstate and the algebraic formulation of \emph{dynamic equilibrium state}, \emph{ground state}, \emph{gapped state}, and \emph{Fermi surface}. Along the way, we provide proofs that are missing or scattered in the literature. 

\medskip

In our opinion, one of the major contributions of this work is the connection between {Fermi surfaces} \cite{Ash,Cal,Kit,Kuc} and eigenstates. In particular we can show (Theorem \ref{prop: eigenstate fermi_S}) that {Fermi surface} can be appropriately interpreted as eigenstates (of the related of the $C^*$-algebra). This paves the way to extend the notion of Fermi surface
beyond the usual setting of periodic operators. In fact in Section \ref{sec:GenFer} we provide a generalized definition of Fermi surface which, in a certain sense, can be imagined as a \virg{non-commutative version} of the usual notion of Fermi surface. This work provides only the first step in this direction, although the desired long-term goal is to have a complete characterization  of the properties of an  eigenstates are to be a \virg{non-commutative Fermi surface}.

\medskip

\noindent
{\bf Structure of the paper and main results.}
In Section \ref{sec:eig-aff} we review  some properties of eigenstates and generalize the results of \cite{rinehart-21} to normal elements. In particular, we show that if $A$ is a normal element of a $C^*$-algebra $\A$, then there is a correspondence between its spectrum and its set of eigenstates (\cf Corollary \ref{corol_norm}). It is also proved that for any normal element the set of eigenstates is closed under functional calculus (\cf Theorem \ref{teoF-calc}). Other characterizations of  eigenstates in terms of   ideals and (left) invertibility are also presented.
Section \ref{sec:eig-din} is devoted to the relation between the notion of eigenstate  of  a self-adjoint element and the algebraic formulation of concepts like \emph{equilibrium state},  \emph{ground state} and \emph{gapped state}. In more detail, we show that  pure states invariant under the dynamic generated by a self-adjoint element $H$ are exactly the eigenstates of $H$ (Proposition \ref{prop:pur_inv}), and we present an explicit characterization of the eigenstates in terms of ground states and gapped  states (Propositions \ref{prop_ground_A} and  \ref{Prop gs}). In Section \ref{secFS}, we introduce a new algebraic definition of  \emph{Fermi surface} for elements of a $C^*$-algebra $\A$. Under mild technical conditions, the main result of this section states that one can build an eigenstate associated with each Fermi surface of a self-adjoint element in  suitable $C^*$-algebras (Theorem \ref{prop: eigenstate fermi_S}). A physical example is is considered in detail in Section \ref{sec:graf}. Supporting material concerning the \emph{disintegration theorem} has been included in Appendix \ref{sec:disint}.

\medskip

\noindent
{\bf Acknowledgements}
GD's research is supported by the grant {Fondecyt Regular - 1190204}. DP’s research is supported by ANID-Subdirección de Capital Humano/ Doctorado Nacional/ 2022-21220144. 
GD would like to cordially thank J. Bellissard   for several inspiring discussions on the algebraic notion of eigenstate and its relation with the notion of Fermi surface, and for suggesting references \cite{riedel-88,riedel-90}.

%--------------------%

\section{Eigenstates of C*-algebras}\label{sec:eig-aff}

%------%
\subsection{Basic definition}\label{sec:def}
In the following, $\rr{A}$ 
will always be a unital $C^*$-algebra with \emph{unit} ${\bf 1}$.
In fact, there is no loss of generality in this assumption, since every non-unital $C^*$-algebra can be endowed with a unit   in a standard way \cite[Proposition 2.1.5]{bratteli-robinson-87}.
The  \emph{state space} of  $\rr{A}$, \ie, the set of the normalized positive linear functionals over $\rr{A}$ \cite[Definition 2.3.14]{bratteli-robinson-87}, 
 will be denoted with
$\n{E}_{\rr{A}}$.
\begin{definition}[Eigenstate]\label{def_00}
Let $A\in\rr{A}$. A state $\omega\in \n{E}_{\rr{A}}$ is an \emph{eigenstate} of $A$ with
\emph{eigenvalue} $\lambda\in\C$ if and only if 
\[
\omega(BA)\;=\;\lambda\;\omega(B)\;,\qquad \forall\; B\in \rr{A}\;.
\]
The set of all eigenvalues of $A$ will be denoted with ${\rm Eig}(A)$.
\end{definition}

\medskip

By choosing $B={\bf 1}$ in the definition above it turns out that 
\[
\lambda\;=\;\omega(A)\;,
\]
namely the  value of the eigenvalue $\lambda$ is determined by the evaluation of the eigenstate $\omega$ on the operator $A$.

\begin{remark}\label{rk:shift-spec}
Let $\omega\in \n{E}_{\rr{A}}$ be an {eigenstate} of $A$ with
{eigenvalue} $\lambda\in\C$. Then $\omega$ is also an an {eigenstate} of $A_0:=A-\lambda_0{\bf 1}$
with {eigenvalue} $\lambda-\lambda_0$ for every $\lambda_0\in\C$.
\hfill $\blacktriangleleft$
\end{remark}

Given $A\in\rr{A}$ and $\lambda\in\C$ let 
\begin{equation}
\rr{J}_{A,\lambda}\;:=\;\overline{\{J=B(A-\lambda{\bf 1})\;|\; B\in\rr{A}\}}\;\subseteq\;\rr{A}\;.
\end{equation}
be  the associated closed left-ideal.

\begin{lemma}\label{lemma_01}
The state $\omega\in \n{E}_{\rr{A}}$ is an  {eigenstate} of $A\in\rr{A}$ with
{eigenvalue} $\lambda\in\C$ if and only if  $\omega|_{\rr{J}_{A,\lambda}}\equiv 0$.
\end{lemma}
\proof
For the implication $(\Rightarrow)$ let us start with  the case $J=B(A-\lambda{\bf 1})$. Then
$$
\omega(J)\;=\;\omega(BA)\;-\;\lambda\omega(B)\;=\;0\;.
$$
In the case $J$ is the norm-limit of the sequence $J_n:=B_n(A-\lambda{\bf 1})$ then one has by continuity that 
$$
|\omega(J)|\;=\;|\omega(J-J_n)|\;\leqslant\;\|J-J_n\|
$$
and in turn $\omega(J)=0$. The implication $(\Leftarrow)$  follows by observing that the condition $\omega|_{\rr{J}_{A,\lambda}}\equiv 0$ implies Definition \ref{def_00}.\qed

\medskip

\begin{proposition}\label{prop:quad=0}
The state $\omega\in \n{E}_{\rr{A}}$ is an  {eigenstate} of $A\in\rr{A}$ with
{eigenvalue} $\lambda\in\C$ if and only if
$$
\omega\big((A-\lambda{\bf 1})^*(A-\lambda{\bf 1})\big)\;=\;0\;.
$$
\end{proposition}
\proof
The implication $(\Rightarrow)$ is a consequence of Lemma \ref{lemma_01}.
The implication $(\Leftarrow)$ follows from the Cauchy-Schwartz inequality for states \cite[Lemma 2.3.10]{bratteli-robinson-87}
$$
\left|\omega\left(B(A-\lambda{\bf 1})\right)\right|^2\;\leqslant\;\omega\left(B^*B\right)\; \omega\left((A-\lambda{\bf 1})^*(A-\lambda{\bf 1})\right)\;=\;0\;,\quad\forall\; B\in\rr{A}
$$
which implies that $\omega(BA)=\lambda\omega(B)$ for every $B\in\rr{A}$.
\qed

\medskip

Let us recall that $A\in\rr{A}$ is \emph{normal} if $AA^*=A^*A$.
\begin{corollary}\label{cor:norm=0}
Let $A\in\rr{A}$ be a normal element and $\omega\in \n{E}_{\rr{A}}$  an eigenstate of $A$ with eigenvalue $\lambda\in\C$. Then $\omega$ is also an  eigenstate of $A^*$ with eigenvalue $\overline{\lambda}\in\C$.
\end{corollary}
\proof In view of Proposition \ref{prop:quad=0} and the normality of $A$ it holds true that
$$
\begin{aligned}
0\;&=\;\omega\big((A-\lambda{\bf 1})^*(A-\lambda{\bf 1})\big)\;=\;\omega\big((A-\lambda{\bf 1})(A-\lambda{\bf 1})^*\big)\\
&=\;\omega\big((A^*-\overline{\lambda}{\bf 1})^*(A^*-\overline{\lambda}{\bf 1})\big)\;.
\end{aligned}
$$
Then, by using again Proposition \ref{prop:quad=0} one concludes the proof.
\qed

%------%
\subsection{Eigenstates and spectrum}\label{sec:eig-spec}
Let us recall that the \emph{spectrum} of $A\in\rr{A}$
is defined as  \cite[Definition 2.2.1]{bratteli-robinson-87}
\[
{\rm Spec}(A)\;:=\;\{\lambda\in\C\;|\; (A-\lambda{\bf 1})\;\text{is not invertible in}\;\rr{A}\}\;.
\]

\medskip

The next result shows that the eigenvalues are elements of the spectrum. 
\begin{proposition}\label{prop_ev_01}
Let $A\in\rr{A}$. Then, ${\rm Eig}(A)\subseteq{\rm Spec}(A)$.
\end{proposition}
\proof
Let $\lambda\in {\rm Eig}(A)$ and  $\omega\in \n{E}_{\rr{A}}$ be the associated eigenstate.
By contradiction, let us suppose that $B=(A-\lambda{\bf 1})^{-1}\in\rr{A}$.
Then, in view of Lemma \ref{lemma_01}, one would have that
$$
\omega({\bf 1})\;=\;\omega\left(B(A-\lambda{\bf 1})\right)\;=\;0
$$
which is a contradiction. Therefore $A-\lambda{\bf 1}$ cannot be  invertible.
\qed

\medskip

It is interesting to have  a criterion that guarantees the equality
\begin{equation}\label{eq:Sp-EV-01}
{\rm Eig}(A)\;=\;{\rm Spec}(A)\;\qquad A\in\rr{A}\;.
\end{equation}
For that, we need the following   preliminary result.
\begin{lemma}\label{lemma_02}
Let $A\in\rr{A}$. Then,
$\lambda\in {\rm Eig}(A)$ if and only if ${\bf 1}\notin \rr{J}_{A,\lambda}$.  
\end{lemma}
\proof
The implication $(\Rightarrow)$ is a direct consequence of Lemma \ref{lemma_01}. In fact, if 
$\omega\in \n{E}_{\rr{A}}$ is an eigenstate of $A$ associated to  $\lambda\in {\rm Eig}(A)$, then $\omega(J)=0$ for every $J\in \rr{J}_{A,\lambda}$, and $\omega({\bf 1})=1$ by definition of states. Consequently ${\bf 1}\notin \rr{J}_{A,\lambda}$. The implication $(\Leftarrow)$ requires the use of 
\cite[Proposition 2.3.24]{bratteli-robinson-87} which is a consequence of  the Hahn-Banach theorem. Let ${\bf 1}\notin \rr{J}_{A,\lambda}$ and consider the unital $C^*$-algebra 
$$
\rr{B}_{A,\lambda}\;:=\;\C{\bf 1}\;+\;\rr{J}_{A,\lambda}\;\subseteq\;\rr{A}\;
$$ 
and  the state $\widetilde{\omega}:\rr{B}_{A,\lambda}\to \C$  defined by
$$
\widetilde{\omega}(\alpha{\bf 1}+J)\;=\;\alpha\;,\qquad \forall\;\alpha\in\C\;,\quad \forall\;J\in\rr{J}_{A,\lambda}\;.
$$
Then, there exists a state $\omega\in \n{E}_{\rr{A}}$ which extends $\widetilde{\omega}$. In particular, this means that 
$\omega|_{\rr{J}_{A,\lambda}}\equiv 0$ and in turn $\lambda\in {\rm Eig}(A)$ in view of Lemma \ref{lemma_01}.
\qed

\medskip

Following \cite[Chapter VII, Definition 3.1]{conway-90} let us introduce the \emph{left spectrum} of $A\in\rr{A}$ defined as
\[
{\rm Spec}_L(A)\;:=\;\{\lambda\in\C\;|\; (A-\lambda{\bf 1})\;\text{is not left-invertible in}\;\rr{A}\}\;.
\]
The \emph{right spectrum} ${\rm Spec}_R(A)$ is defined similarly.
From the definitions above, it holds true that
$$
{\rm Spec}(A)\;=\;{\rm Spec}_L(A)\;\cup\; {\rm Spec}_R(A)\;.
$$
Moreover one has that ${\rm Spec}_L(A)={\rm Spec}_R(A^*)$\;.

\medskip

\begin{remark}\label{rk:L-R-spec}
Let $\rr{A}=\rr{B}(\s{H})$ be the  $C^*$-algebra of bounded operators over some Hilbert space $\s{H}$ and $A\in\rr{A}$. Then, in view of \cite[Chapter XI, Proposition 1.1]{conway-90} one has that
${\rm Spec}_L(A)$ 
and ${\rm Spec}_R(A)$ coincide with the \emph{approximate point spectrum} and the \emph{surjective spectrum} of $A$, respectively .
\hfill $\blacktriangleleft$
\end{remark}

\medskip

\begin{proposition}\label{prop_ev_02}
Let $A\in\rr{A}$. Then, ${\rm Spec}_L(A)={\rm Eig}(A)$.
\end{proposition}
\proof
The inclusion ${\rm Eig}(A)\subseteq {\rm Spec}_L(A)$ is a direct consequence of Lemma \ref{lemma_02}. To prove the opposite inclusion let $\lambda\in {\rm Spec}_L(A)$. In order to show that $\lambda\in{\rm Eig}(A)$ it is enough to prove that ${\bf 1}\notin \rr{J}_{A,\lambda}$ in view of Lemma \ref{lemma_02}.
Let $\rr{J}'_{A,\lambda}:=\rr{A}(A-\lambda{\bf 1})$ be the dense subideal of $\rr{J}_{A,\lambda}$ defined by elements of the form 
$J=B(A-\lambda{\bf 1})$ with $B\in \rr{A}$. Since $A-\lambda{\bf 1}$ is not left-invertible by definition, it follows that ${\bf 1}\notin \rr{J}'_{A,\lambda}$. 
Therefore no element of $\rr{J}'_{A,\lambda}$ can be invertible.
Thus, for every $J\in  \rr{J}'_{A,\lambda}$ it holds true that $0\in {\rm Spec}(J)$, or equivalently $1\in {\rm Spec}({\bf 1}-J)$. It follows that $\|J-{\bf 1}\|\geqslant 1$, and consequently
 ${\bf 1}\notin \rr{J}_{A,\lambda}$ by a continuity argument.
\qed

\medskip

Putting together the content of Propositions \ref{prop_ev_01} and  \ref{prop_ev_02} one gets
\begin{equation}\label{eq:Sp-EV-02}
{\rm Spec}_L(A)\;=\;{\rm Eig}(A)\;\subseteq\;{\rm Spec}(A)\;\qquad A\in\rr{A}\;.
\end{equation}
In order to pass from \eqref{eq:Sp-EV-02} to \eqref{eq:Sp-EV-01}
one needs some more condition that guarantees equality between  the left spectrum and the spectrum. 

\medskip

\begin{corollary}\label{corol_norm}
For every normal element $A\in\rr{A}$ the equality \eqref{eq:Sp-EV-01} holds true.
\end{corollary}
\proof
For normal elements, the equality between  the left spectrum and spectrum is proved in \cite[Chapter XI, Proposition 1.4]{conway-90}. Then the result follows from \eqref{eq:Sp-EV-02}.
\qed

\medskip

Let $\rr{K}(\s{H})$ be the (non-unital) $C^*$-algebra of compact operators on the separable infinite-dimensional  Hilbert space $\s{H}$. As usual, in order to have a unit let us consider the standard  extension $\rr{K}^+(\s{H}):=\C{\bf 1}+\rr{K}(\s{H})$.
Since the structure of $\rr{K}(\s{H})$ does not depend on the specific (separable)  space $\s{H}$, we will use the short notations $\rr{K}$ and $\rr{K}^+$.
\medskip

\begin{corollary}
For every  $A\in\rr{K}^+$ the equality \eqref{eq:Sp-EV-01} holds true.
\end{corollary}
\proof
Let $A=T+z{\bf 1}\in \rr{K}^+$ for  some $z\in \mathbb{C}$ and $T\in \rr{K}$. In view of Proposition \ref{prop_ev_01} we only need to prove 
that ${\rm Spec}(A)\subseteq{\rm Eig}(A)$. Let
 $\lambda \in {\rm Spec}(A) $.  Then $T-(\lambda-z){\bf 1}$  is not invertible,  hence $\lambda-z\in {\rm Spec}(T)$. If $\lambda\neq z$, then by the Riesz-Schauder theorem \cite[Theorem VI.15]{reed-simon-I}  there exists a $\psi\neq 0$ in $\mathcal{H}$ such that $T\psi=(\lambda-z)\psi$. Therefore,  $A\psi=\lambda \psi$ and $\lambda \in  {\rm Eig}(A)$ as we will see in Theorem \ref{teo_rep}.
 Now, let $\lambda=z$ and assume that $\lambda \notin {\rm Eig}(A)$. Then, from Proposition \ref{prop_ev_02} one infers the existence of a bounded operator $S$ such that
    $$
    {\bf 1}\;=\;S(A-\lambda{\bf 1})\;=\;ST\;.
    $$
The last equality would 
    imply that  the identity ${\bf 1}$ is compact since $\rr{K}$ is an ideal. However, this is a  contradiction whenever $\s{H}$ is infinite-dimensional.
    This completes the proof.
\qed

\begin{example}[Unilateral shift]\label{ex:uni-shif}
It is not hard to check that there are operators for which the equality \ref{eq:Sp-EV-01} does not hold. Let us
consider the   unilateral shift operator $\rr{s}$ on $\ell^2(\mathbb{N})$ defined by
$$\rr{s}\;:\;(n_1,n_2,...)\;\longmapsto\;(0,n_1,n_2,...)\;,\qquad\forall\;n\;:=\;(n_1,n_2,...)\in \ell^2(\mathbb{N})\;.
$$
Its adjoint is given by $\rr{s}^*:(n_1,n_2,...)\mapsto(n_2,n_3,...)$.
From \cite[Chapter VII, Proposition 6.5 \& Corollary 6.6]{conway-90}
one knows that ${\rm Spec}(\rr{s})={\n{D}_1}={\rm Spec}(\rr{s}^*)$ where $\n{D}_1:=\{\lambda\in\C\;|\;|\lambda|\leqslant1\}$ is the closed unit disk.
On the other hand, by   \cite[Chapter XI, Proposition 1.1]{conway-90}  one gets that ${\rm Spec}_L(\rr{s})={\rm Spec}_{a.p.}(\rr{s})=\partial \n{D}_1\equiv\n{S}^1$ 
where  $\n{S}^1:=\{\lambda\in\C\;|\;|\lambda|=1\}$. In view of Proposition
\ref{prop_ev_02}
  one finally obtains that ${\rm Eig}(\rr{s})\simeq\n{S}^1$ showing that ${\rm Eig}(\rr{s})\neq {\rm Spec}(\rr{s})$. It is worth to point out that the operator $\rr{s}^*$ meets equation \eqref{eq:Sp-EV-01} since it holds true that  ${\rm Spec}_L(\rr{s}^*)={\n{D}_1}$.
\hfill $\blacktriangleleft$
\end{example}

\medskip

The next result uses the functional calculus for normal elements of $A\in\rr{A}$. In this case $A$ and $A^*$  
generate a commutative sub-$C^*$-algebra of $\rr{A}$ which is isomorphic to the $C^*$-algebra $C({\rm Spec}(A))$ of continuous functions over the compact set ${\rm Spec}(A)\subset\C$ in view of the {Gelfand-Nainmark theorem} \cite[Section VII.2]{conway-90}. In particular, for every continuous function $f\in C({\rm Spec}(A))$ there is an associated (normal) element $f(A)\in \rr{A}$ with spectrum $f({\rm Spec}(A))$ (spectral mapping theorem). From Corollary \ref{corol_norm} one gets that 
\begin{equation}\label{eq:Sp-EV-0101}
{\rm Eig}(f(A))\;=\;f({\rm Spec}(A))\;.
\end{equation}
The next result relates the eigenstates of $A$ with the eigenstates of $f(A)$ when $A$ is a normal element.

\begin{theorem}\label{teoF-calc}
Let $A\in\rr{A}$ be a normal element and $f\in C({\rm Spec}(A))$. If $\omega$ is an eigenstate of $A$ with eigenvalue $\lambda\in\C$ then $\omega$ is also an eigenstate of $f(A)$ with eigenvalue $f(\lambda)\in\C$.
\end{theorem}
\proof
By using induction on $n$ and $m$ and Corollary \ref{cor:norm=0} one obtains for any monomial $p_{n,m}(x)\;:=\;x^n\overline{x}^m$ that
$$\omega(Bp_{n,m}(A))\;=\;\omega(BA^n(A^*)^m)\;=\;\lambda^n\overline{\lambda}^m\omega(B)\;=\;p_{n,m}(\lambda)\omega(b).$$
Then, by linearity one gets that $\omega$ is   an eigenstate of $p(A)$ with  eigenvalue $p(\lambda)$ for every polynomial $p$.
The final result follows from the Stone-Weierstrass theorem and the continuity of $\omega$.
\qed

\medskip

\begin{corollary}\label{cor_commut}
Let $A\in \rr{A}$ be a normal 
and $\omega\in \n{E}_{\rr{A}}$ an eigenstate of $A$ related to the eigenvalue $\lambda\in{\rm Spec}(A)$. Let $f:\C\to\C$ be a continuous function. Then,
$$
\omega([B,f(A)])\;=\;\omega(Bf(A))\;-\; \omega(f(A)B)\;=\;0
$$
for every $B\in \rr{A}$.
\end{corollary}
\proof
From Theorem \ref{teoF-calc} one has that 
$$
\omega(Bf(A))\;=\;f(\lambda)\; \omega(B)\;.
$$
On the other hand, by Corollary \ref{cor:norm=0} it is also true that
$$
\begin{aligned}
\omega(f(A)B)\;&=\;\overline{\omega\left(B^*f(A)^*\right)}\;=\;\overline{\overline{f(\lambda)}\; \omega\left(B^*\right)}\\
&=\;f(\lambda)\;\overline{\omega\left(B^*\right)}\;=\;f(\lambda)\; \omega(B) \;.
\end{aligned}
$$
This completes the proof.
\qed

%------%
\subsection{Eigenstates and representations}\label{sec:eig-rep}
In this section, we will study the behavior of the notion of eigenstate under $\ast$-representations of the $C^*$-algebra $\rr{A}$ 
in the algebra $\rr{B}(\s{H})$ of bounded operators on the Hilbert space $\s{H}$.

\begin{theorem}\label{teo_rep}
Let $\pi:\rr{A}\to\rr{B}(\s{H})$ be a $\ast$-representation of $\rr{A}$ and
$$
\omega_\psi(A)\;:=\;\langle\psi,\pi(A)\psi\rangle_{\s{H}}\;,\qquad A\in\rr{A}
$$
be the vector state associated with the normalized vector $\psi\in\s{H}$.
Then, $\omega_\psi$ is an eigenstate of $A$ with eigenvalue $\lambda\in\C$ if and only if $\pi(A)\psi=\lambda\psi$.
\end{theorem}
\proof
Let us start with the implication $(\Rightarrow)$. If 
$\omega_\psi$ is an eigenstate then, in view of Lemma \ref{lemma_01}, one obtains
$$
0\;=\;\omega_\psi\left((A-\lambda{\bf 1})^*(A-\lambda{\bf 1})\right)\;=\;\|(\pi(A)-\lambda{\bf 1})\psi\|_{\s{H}}^2
$$
which implies $\pi(A)\psi=\lambda\psi$. 
The implication  $(\Leftarrow)$ follows from the direct computation
$$
\omega_\psi(BA)\;=\;\langle\psi,\pi(B)\pi(A)\psi\rangle_{\s{H}}\;=\;\lambda\;\langle\psi,\pi(B)\psi\rangle_{\s{H}}\;=\;\lambda\;
\omega_\psi(B)\;.
$$
This completes the proof.
\qed

\medskip

Let us recall that every state $\omega\in \n{E}_{\rr{A}}$ define the \emph{GNS representation} $(\pi_\omega, \s{H}_\omega, \psi_\omega)$
\cite[Chapter VIII, Theorem 5.14]{conway-90}. Here $\psi_\omega$ is the cyclic vector of the representation, and it holds true that
$$
\omega(A)\;=\;\langle\psi_\omega,\pi_\omega(A)\psi_\omega\rangle_{\s{H}_\omega}\;,\qquad A\in\rr{A}\;.
$$
The following result is an immediate consequence of Theorem \ref{teo_rep}.
\begin{corollary}\label{coro 2.15}
Let $A\in\rr{A}$ and $\omega\in \n{E}_{\rr{A}}$  an eigenstate of $A$ with eigenvalue $\lambda\in\C$. Let $(\pi_\omega, \s{H}_\omega, \psi_\omega)$
 be the GNS representation of $\omega$. Then
$$
\pi_\omega(A)\psi_\omega\;=\;\lambda\psi_\omega\;.
$$
\end{corollary}
%------%
\subsection{Relation between distinct eigenstates}\label{sec:eig-dist}
In this section, we will investigate the relation between eigenstates of the element $A\in\rr{A}$ related to distinct eigenvalues. Let us start with a definition. A collection of states $\{\omega_1,\ldots,\omega_N\}\in \n{E}_{\rr{A}}$ is called \emph{linearly independent} if
 $\sum_{k=1}^Na_k\omega_k\equiv0$ with $a_j\in\C$ implies 
$a_1=\ldots=a_N=0$.
\begin{theorem}
Any collection $\{\omega_1,\ldots,\omega_N\}\in \n{E}_{\rr{A}}$ of eigenstates of $A\in\rr{A}$ all having distinct eigenvalues $\{\lambda_1,\ldots,\lambda_N\}\subset\C$, is
linearly independent.
\end{theorem}
\proof
The proof is by  induction on $N$. The case $N=1$ is clear.
Assume now that the conclusion holds for $N-1$, and
suppose that $\sum_{k=1}^Na_k\omega_k\equiv0$. By applying this sum to the element $B(A-\lambda_N{\bf 1})$ for some $B\in \rr{A}$  one gets
$$
\begin{aligned}
0\;&=\;\sum_{k=1}^Na_k\omega_k\big(B(A-\lambda_N{\bf 1})\big)\\&=\;\sum_{k=1}^{N-1}a_k\omega_k\big(B(A-\lambda_N{\bf 1})\big)
\;=\;\sum_{k=1}^{N-1}a_k(\lambda_k-\lambda_N)\omega_k(B)\;.
\end{aligned}
$$
By the inductive hypothesis one has that $a_k(\lambda_k-\lambda_N)=0$ for every $k=1,\ldots,N-1$. Since all the eigenvalues are distinct it follows that $a_1=\ldots=a_{N-1}=0$. Finally $a_N=0$ in view of  the argument for  $N=1$.
\qed

\medskip

Following \cite[Definition 3.2.3]{pedersen-79},
let us recall that two linear functionals $\omega_1$ and $\omega_2$
over $\rr{A}$ are said \emph{orthogonal}
if
$$
\|\omega_1-\omega_2\|\;=\;\|\omega_1\|\;+\; \|\omega_2\|\;.
$$

\begin{theorem}
Let $A\in \rr{A}$ be normal. Then, any two eigenstates of $A$ with different eigenvalues are orthogonal.
\end{theorem}
\proof
Let $\omega_1$ and $\omega_2$ be two eigenstates of $A$ with eigenvalues $\lambda_1\neq\lambda_2$, respectively. As a consequence of the Urysohn's lemma, there exists a continuous function $f\in C({\rm Spec}(A))$ such that $f(\lambda_1)=1$, $f(\lambda_2)=-1$ and $\|f\|_\infty=1$. From Theorem \ref{teoF-calc} one gets that 
$$
(\omega_1-\omega_2)(f(A))\;=\;f(\lambda_1)\;-\;f(\lambda_2)\;=\;2\;.
$$
Consequently,
$$
2\;\leqslant\;\|\omega_1-\omega_2\|\;\leqslant\;\|\omega_1\|\;+\;\|\omega_2\|\;=\;2
$$
which implies the orthogonality between $\omega_1$ and $\omega_2$.
\qed

%------%
\subsection{Eigenstates of a projection}\label{sec:eig-proj}

Let us recall that a self-adjoint projection is an element 
$P\in\rr{A}$ such that $P^*=P=P^2$.

\medskip

\begin{theorem}\label{eigen:projection}
Let $P\in\rr{A}$ be a self-adjoint projection and $\omega\in \n{E}_{\rr{A}}$. Let $\omega_P\in \n{E}_{\rr{A}}$ be the state defined by
$$
\omega_P(A)\;:=\;\frac{\omega(PAP)}{\omega(P)}\;,\qquad \forall\; A\in\rr{A}\;.
$$
The following are equivalent:
\begin{itemize}
\item[(i)] $\omega_P=\omega$;
\vspace{1mm}
\item[(ii)] $\omega$ is an eigenstate of $P$ with eigenvalue $1$;
\vspace{1mm}
\item[(iii)] $\omega(P)=1$.
\end{itemize}
\end{theorem}
\proof
$(i)\Leftrightarrow (ii)$. For all $B\in\rr{A}$ one has that
$$
\omega(BP)\;=\;\omega_P(BP)\;=\;\frac{\omega(PBP^2)}{\omega(P)}\;=\;\frac{\omega(PBP)}{\omega(P)}\;=\;\omega_P(B)\;=\;\omega(B)\;,
$$
then $\omega$ is an eigenstate of $P$ with eigenvalue $1$. On the other hand, if $\omega$ is an eigenstate of $P$ with eigenvalue $1$, then $\omega(BP)=\omega(B)$ for every $B\in\rr{A}$,
and in turn  $\omega(PB^*)=\omega(B^*)$ by taking the adjoint.
By replacing the generic element $B^*$ with $BP$ one gets
 $\omega(PBP)=\omega(BP)=\omega(B)$. This fact along with $\omega(P)=1$ provides $\omega_P(B)=\omega(B)$ for every $B\in\rr{A}$.\\
$(ii)\Leftrightarrow (iii)$. If $\omega$ is an eigenstate of $P$ with eigenvalue $1$ then
$$
0\;=\;\omega(0)\;=\;\omega(({\bf 1}-P)P)\;=\;\omega({\bf 1}-P)\;.
$$
Therefore
$$
1\;=\;\omega({\bf 1})\;=\;\omega(({\bf 1}-P)+P)\;=\;\omega(P)\;.
$$
On the other hand from $\omega(P)=1$ one gets $\omega({\bf 1}-P)=0$ and in turn $\omega(({\bf 1}-P)^*({\bf 1}-P))=\omega({\bf 1}-P)=0$
which implies that $\omega$ is an eigenstate of $P$ with eigenvalue $1$ in view of Proposition \ref{prop:quad=0}.\qed

\begin{remark}[Eigenstates for gapped spectral projections]
Let $H=H^*$ be a selfadjoint element in $\mathfrak{A}$ and $\sigma_*\subset {\rm Spec}(H)\subset\R$ a closed (hence compact) subset of the spectrum of $H$. 
Let us assume the \emph{gap condition}
 ${\rm dist}(\sigma_*,{\rm Spec}(H)\setminus \sigma_*)>0$ and denote with $P_*$  the spectral projection of $H$ on $\sigma_*$. Via functional calculus,  and in view of the gap condition, we can choose a real continuous function $f:\mathbb{R}\rightarrow\mathbb{R}$ such that $f$ coincides with the characteristic function of $\sigma_\ast$ when restricted to ${\rm Spec}(H)$,  and  $f(H)=P_*$. In particular, this ensures that $P_\ast\in \rr{A}$. 
 We will refer to $P_*$ as a \emph{gapped} spectral projection of $H$. Our aim is to construct an eigenstate of $P_*$ (with eigenvalue $1$). By invoking the disintegration theorem for measures (\cf Theorem \ref{teo: desintegration}),  one can obtain a real probability measure $\mu_1$ with support in $\sigma_*$. In fact, let $\nu:=\mu\circ f^{-1}$ be the pushforward measure by $f$ of the standard Lebesgue measure $\mu$ on $\mathbb{R}$. By  Theorem \ref{teo: desintegration} there exists a $\nu$-almost everywhere uniquely determined disintegration $\{\mu_t\}_{t\in \mathbb{R}}$ of $\mu$ over $\nu$. Since  $\mu_t$ lives on the fiber $f^{-1}(\{t\})$ for 
$\nu$-almost all $t\in\R$, it follows that $\mu_1$ is supported on
 $f^{-1}(\{1\})=\sigma_*$ as required. With  $\mu_1$ we can generate  an eigenstate $\omega$ of $P_*$. Indeed, from Corollary \ref{corol_norm}, for each $\lambda\in \sigma_*$ there exists at least an eigenstate $\omega_\lambda$ of $H$ with eigenvalue $\lambda$. 
 Let us assume that there is a point-wise continuous (or measurable) map $\sigma_*\ni\lambda\mapsto\omega_\lambda$. Then  the functional
\[\omega(A)\;=\;\int_{\sigma_*} \,d\mu_1(\lambda)\,\omega_\lambda(A)\;,\hspace{1cm}A\in \mathfrak{A} \]
results well-defined. Moreover, one can check that $\omega(1)=\mu_1(\sigma_*)=1$, which implies that $\omega\in {\rm E}_\A$. Finally, a direct computation shows
\begin{equation*}
    \begin{split}
        \omega(P_*)\;&=\;\int_{\sigma_*}{\rm d}\mu_1(\lambda)\,\omega_\lambda(P_*)\;=\;\int_{\sigma_*}d\mu_1(\lambda)\,f(\lambda)\;=\;\int_{\sigma_*}{\rm d}\mu_1(\lambda)\\
        &=\;\mu_1(\sigma_*)\;=\;1
    \end{split}
\end{equation*}
and by Theorem \ref{eigen:projection}, $\omega$ is an eigenstate of $P_*$ with eigenvalue $1$. \hfill $\blacktriangleleft$
\end{remark}

\medskip

The construction above will be taken up, and generalized in a  larger context in Section \ref{asec:class_eig}.

%--------------------%
%--------------------%

\section{Eigenstates and dynamics}\label{sec:eig-din}
In this section, we will relate the concept of eigenstates of a self-adjoint operator $H$ with certain properties of the dynamics induced by $H$.

%------%
\subsection{Stability under the dynamics}\label{sec:stab-din}
Let $H=H^*$ be a self-adjoint element of the $C^*$-algebra $\rr{A}$. By functional calculus, the unitaries $\expo{\ii t H}$, with $t\in\R$, are contained in $\rr{A}$. Then, one can define the \emph{dynamics} associated with $H$ as the one-parameter group of automorphisms $t\mapsto \alpha^H_t\in{\rm Aut}(\rr{A})$ defined by
\begin{equation}\label{eq:dyn_def}
\alpha^H_t(A)\;:=\;\expo{\ii t H}\; A\; \expo{-\ii t H}\;,\qquad \forall\; A\in\rr{A}\;.
\end{equation}
The dynamics $\alpha^H$ defined above is \emph{strongly continuous} in the sense that the functions $t\mapsto\|\alpha^H_t(A)\|$ are continuous for every $A\in\rr{A}$. This easily follows from the boundedness of $H$ which guarantees the continuity of $t\mapsto\|\expo{\ii t H}\|$. It turns out that the dynamics are also strongly differentiable in the sense that the limit
$$
\delta_H(A)\;:=\;\ii [H,A]\;=\;\lim_{t\to0}\frac{\alpha^H_t(A)- A}{t}
$$
exists in norm for every $A\in\rr{A}$. The commutator $\delta_H$ acts on $\rr{A}$ as a (bounded) derivation and it is called the \emph{infinitesimal generator} of the {dynamics} $\alpha^H$.

\medskip

A state $\omega\in \n{E}_{\rr{A}}$ is \emph{invariant} under the dynamics  induced by $H$ if and only if 
$$
\omega\circ \alpha^H_t\;=\;\omega\;,\qquad \forall\; t\in\R\;.
$$ 
This is equivalent to $\omega\circ \delta_H=0$. In view of of Corollary \ref{cor_commut} one has that:
\begin{proposition}\label{prop-inv}
Let $H=H^*$ be a self-adjoint element of the $C^*$-algebra $\rr{A}$
and $\omega\in \n{E}_{\rr{A}}$ an eigenstate of $H$ related to the eigenvalue $\lambda\in{\rm Spec}(H)$. Then $\omega$ is invariant under the dynamics  induced by $H$. 
\end{proposition}
\begin{remark}\label{remark 3,2}
It is not hard to show that the converse of Proposition \ref{prop-inv} is not true in general. In fact, there are invariant states which are not eigenstates. Consider the commutative $C^*$-algebra $C([0,1])$,  the state given by the integral
$$
\omega(f)\;:=\;\int_0^1 \dd s\; f(s)\;,\qquad f\in C([0,1])\;,
$$
and the   element $\xi(s):=s$ for all $s\in[0,1]$. Since 
$\omega(\xi^n)=(n+1)^{-1}$ for every   $n\in\N\cup\{0\}$, one gets that 
$\omega (\xi)=\frac{1}{2}$ and $\omega (\xi^2)\neq \frac{1}{2}\omega(\xi)$. Therefore, $\omega$ is not an eigenstate for $\xi$ (associated to $\lambda=\frac{1}{2}$). Nevertheless, in view of the commutativity of $C([0,1])$,  $\omega$ is trivially invariant under the dynamics induced by $\xi$. \hfill $\blacktriangleleft$
\end{remark}

In order to obtain the converse of Proposition \ref{prop-inv} we need to require some more conditions on the state $\omega$. Let us recall that a state $\omega$ is \emph{pure} if and only if the associated 
GNS representation $
( \pi_\omega,\mathcal{H}_\omega,\psi_\omega)$ is irreducible \cite[Theorem 2.3.19]{bratteli-robinson-87}. Recall that in the commutative case, the pure states coincide with the multiplicative functionals \cite[Corollary 2.3.21]{bratteli-robinson-87}, and therefore they are automatically eigenstates.

\begin{proposition}\label{prop:pur_inv}
Let $H=H^*$ be a self-adjoint element of the $C^*$-algebra $\rr{A}$, $t\mapsto \alpha^H_t$ the dynamics generated by $H$ according to \eqref{eq:dyn_def} and   $\omega\in \n{E}_{\rr{A}}$ a pure state. If $\omega$ is invariant under the dynamics $\alpha^H_t$ then it is an eigenstate of $H$.
\end{proposition}
\proof
Let  $(\pi_\omega, \s{H}_\omega, \psi_\omega)$
 be the GNS representation associated with $\omega$. As a consequence of the uniqueness up to unitary equivalence of the GNS representation  \cite[Corollary  2.3.17]{bratteli-robinson-87} one has that it exists a unique strongly continuous one-parameter group of  unitary operators $U_\omega(t)$ on $\s{H}_\omega$ such that: (i) $U_\omega(t)\pi_\omega(A)\psi_\omega=\pi_\omega(\alpha^H_t(A))\psi_\omega$, and (ii) $U_\omega(t)\psi_\omega=\psi_\omega$ for every $A\in\rr{A}$ and $t\in\R$. 
 By the Stone’s theorem there exists a self-adjoint operator $K_\omega$ on $\s{H}_\omega$ such that $U_\omega(t)=\expo{\ii t K_\omega}$ for every $t\in \R$, and in turn 
  $K_\omega\psi_\omega=0$ in view of the property (ii). On the other hand, since $\pi_\omega$ is a representation and $H\in\rr{A}$ one has that
 \[
 \pi_\omega(\alpha^H_t(A))\;:=\;\expo{\ii t \pi_\omega(H)}\; \pi_\omega(A)\; \expo{-\ii t \pi_\omega(H)}\;.
 \] 
  Therefore $K_\omega$ and $\pi_\omega(H)$ generate the same dynamics on $\s{H}_\omega$, and in turn the same commutator.
  In other words one has
  \[
  [\pi_\omega(H)-K_\omega,\pi_\omega(A)]\;=\;0\;,\qquad\forall A\in\rr{A}\;.
  \]
By considering the irreducibility of  the representation one gets  $\pi_\omega(H)-K_\omega=\lambda{\bf 1}$ for some $\lambda\in\C$ in view of the Schur's lemma. It turns out that $\pi_\omega(H)\psi_\omega=(K_\omega+\lambda{\bf 1})\psi_\omega=\lambda\psi_\omega$ and this complete the proof in view of 
  Corollary \ref{coro 2.15}.
  \qed

%------%
\subsection{Ground state condition}\label{sec:gs}
In this section, we will introduce the concept of 
ground state and we will study its relation with the dynamics.

\begin{definition}[Ground state]\label{deg:gs}
Let $H=H^*$ be a self-adjoint element of the $C^*$-algebra $\rr{A}$,  $\omega\in \n{E}_{\rr{A}}$ a  state and  $(\pi_\omega, \s{H}_\omega, \psi_\omega)$
 the related GNS representation. Then, $\omega$ is a \emph{ground state} for $H$ if and only if 
 $$
 \omega(H)\;=\:\lambda_*\;:=\; \min\; {\rm Spec}(\pi_\omega(H))\;.
 $$
\end{definition}

\medskip

The next result shows that a ground state is automatically an eigenstate.

\begin{proposition}\label{prop:gs}
Let $H=H^*$ be a self-adjoint element of the $C^*$-algebra $\rr{A}$ and $\omega\in \n{E}_{\rr{A}}$ a {ground state} for $H$. Then, $\omega$ is an eigenstate of $H$ related to the eigenvalue $\lambda_*$.
\end{proposition}
\proof
In view of Theorem \ref{teo_rep}, it is sufficient to prove that
\begin{equation}\label{eq:gs-p1}
\pi_\omega(H)\psi_\omega\;=\;\lambda_*\psi_\omega
\end{equation}
 where $( \pi_\omega,\mathcal{H}_\omega,\psi_\omega)$ is the GNS representation associated to $\omega$. Let us assume that ${\rm Spec}(\pi_\omega(H))\setminus \{\lambda_*\}\neq\emptyset$ 	otherwise it is trivial.  Let $E^H$ be the PVM associated to $\pi_\omega(H)$ and 
$$
\pi_\omega(H)\;=\;\int_{{\rm Spec}(\pi_\omega(H))} \dd E^H(\lambda)\;\lambda 
$$
its spectral decomposition. Then
$$
\begin{aligned}
0\;&=\;\omega(H)\;-\;\lambda_*\\
&=\;\langle \psi_\omega,(\pi_\omega(H)-\lambda_*{\bf 1})\psi_\omega\rangle_{\mathcal{H}_\omega}\\
&=\;\int_{{\rm Spec}(\pi_\omega(H))}  \langle \psi_\omega,\dd E^H(\lambda)\psi_\omega\rangle_{\mathcal{H}_\omega}\; (\lambda-\lambda_*)\\
&=\;\int_{{\rm Spec}(\pi_\omega(H))\setminus\{\lambda_*\}}  \langle \psi_\omega,\dd E^H(\lambda)\psi_\omega\rangle_{\mathcal{H}_\omega}\; (\lambda-\lambda_*)\;.
\end{aligned}
$$
Since the difference $\lambda-\lambda_*$ is strictly positive on the set ${\rm Spec}(\pi_\omega(H))\setminus\{\lambda_*\}$, it follows that
$$
\langle \psi_\omega, E^H({\rm Spec}(\pi_\omega(H))\setminus\{\lambda_*\})\psi_\omega\rangle_{\s{H}_\omega}\;=\;0
$$
and in turn
$$
\langle \psi_\omega, E^H(\{\lambda_*\})\psi_\omega\rangle_{\s{H}_\omega}\;=\;
\langle \psi_\omega, E^H({\rm Spec}(\pi_\omega(H)))\psi_\omega\rangle_{\s{H}_\omega}\;=\;\langle \psi_\omega, \psi_\omega\rangle_{\s{H}_\omega}\;=\;1\;.
$$
Since $E^H(\{\lambda_*\})$ is an orthogonal projection  one gets that $\|E^H(\{\lambda_*\})\psi_\omega\|_{{H}_\omega}=1$. Then the Pythagoras theorem implies that
$$
E^H(\{\lambda_*\})\psi_\omega\;=\;\psi_\omega\;.
$$
The latter equality and the spectral decomposition of $H$  imply  the relation \eqref{eq:gs-p1}. 
\qed

\begin{remark}[Absolute ground state]
The state $\omega\in \n{E}_{\rr{A}}$ will be called an \emph{absolute ground state} of $H$ if 
$$
 \omega(H)\;=\: \lambda_0\;:=\; \min\; {\rm Spec}(H)\;.
 $$
Since ${\rm Spec}(\pi_\omega(H))\subseteq {\rm Spec}(H)$ one gets that $\lambda_0\leqslant \min {\rm Spec}(\pi_\omega(H))$. However,
the  same argument in the proof of Proposition \ref{prop:gs} shows  that $\lambda_0$ is the eigenvalue of $\pi_\omega(H)$ related to the cyclic vector $\psi_\omega$.
  Therefore, one gets 
   that $\lambda_0=\min {\rm Spec}(\pi_\omega(H))$ meaning that an {absolute} ground state is  a ground state according to Definition \ref{deg:gs}.
\hfill $\blacktriangleleft$
\end{remark}

\medskip

The concept of ground states  can be expressed in terms of a certain dynamical condition as in 
\cite[Definition  5.3.18]{bratteli-robinson-97}.
The next characterization of ground states is an adaption of \cite[Proposition  5.3.19]{bratteli-robinson-97}. 
\begin{proposition}\label{prop_ground_A}
Let $H=H^*$ be a self-adjoint element of the $C^*$-algebra $\rr{A}$. 
If $\omega\in \n{E}_{\rr{A}}$ is a ground state for $H$ then
\begin{equation}\label{ineq_groun}
-\ii\;\omega\big(A^*\delta_H(A)\big)\;\geqslant\;0\;, \qquad \forall\; A\in\rr{A}\;.  
\end{equation}
On the other hand, if  $\omega\in \n{E}_{\rr{A}}$ is a pure state which meets condition \eqref{ineq_groun} then $\omega$ is a ground state for $H$.
\end{proposition}
\proof
Let $( \pi_\omega,\mathcal{H}_\omega,\psi_\omega)$ be the GNS representation associated with the state $\omega$. For the first implication, we have that
\begin{equation*}
    \begin{split}
        -\ii\;\omega\big(A^*\delta_H(A)\big)\;&=\;\omega\big(A^*HA-A^*AH\big)\;=\;\omega\big(A^*HA\big)-\omega\big((A^*A)H\big) 
        \\
        &=\;\omega\big(A^*HA\big)-\lambda_*\,\omega\big(A^*A\big)\;=\;\omega\big(A^*(H-\lambda_*{\bf 1}
        )A\big)\\
        &=\;\langle \psi_\omega\,,\, \pi_\omega(A)^*(\pi_\omega(H)-\lambda_*{\bf 1}
        )\pi_\omega(A)\psi_\omega\rangle\\
      &  =\;\langle \psi_A\,,\,(\pi_\omega(H)-\lambda_*{\bf 1}
        )\psi_A\rangle\;\geq\; 0\;,
    \end{split}
\end{equation*}
where in the second line we  used Proposition \ref{prop:gs} which ensures that $\omega$ is an eigenstate of $H$ with eigenvalue $\lambda_\ast$. In the last line, we introduced
$\psi_A:=\pi_\omega(A)\psi_\omega$, and the last inequality is a consequence of the fact that $\lambda_*$ is by assumption  the minimum of the spectrum of $\pi_\omega(H)$. For the second implication, let us start by observing that \cite[Lemma 5.3.16]{bratteli-robinson-97} implies $\omega\circ \delta_H=0$, \ie  
$\omega$  is invariant under the dynamics induced by $H$. Assuming that $\omega$ is pure one can use Proposition \ref{prop:pur_inv} which shows that $\omega$ is an eigenstate of $H$
with eigenvalue $\lambda_\ast:=\omega(H)$.
With this result
inequality 
\eqref{ineq_groun} reads $\lambda_*\omega(A^*A)\leqslant\omega(A^*HA)$ for every $A\in \rr{A}$. The latter implies
$$
\lambda_*\;\leqslant\;\inf_{\substack{A\in \rr{A} \\ \omega(A^*A)>0}} \frac{\omega(A^*HA)}{\omega(A^*A)}\;=\;\inf_{\substack{A\in \rr{A} \\ \omega(A^*A)>0}} \frac{\langle\psi_A,\pi_\omega(H)\psi_A\rangle_{\s{H}_\omega}}{\|\psi_A\|^2_{\s{H}_\omega}}\;.
$$
 In view of the cyclicity of the vector $\psi_\omega$ one finally gets
$$
\lambda_*\;\leqslant\;\inf_{ \|\phi\|_{\s{H}_\omega}=1} \langle\phi,\pi_\omega(H)\phi\rangle_{\s{H}_\omega}\;=\;\inf {\rm Spec}\big(\pi_\omega(H)\big)
$$
where the last inequality is justified by  \cite[Theorem 2.19]{Tes}. Since $\lambda_*$ belongs to the spectrum of $\pi_\omega(H)$ one concludes that $\lambda_*=\min {\rm Spec}\big(\pi_\omega(H)\big).$
\qed

%------%
\subsection{Gapped ground states}\label{sec:gap}
In many applications it is relevant to know when the ground state of a system is protected by a gap from the rest of the spectrum. Let us introduce the following precise definition.

\begin{definition}[Gap condition]
Let  $H=H^*$ be a self-adjoint element of the $C^*$-algebra $\rr{A}$ and $\omega\in \n{E}_{\rr{A}}$ a ground state of $H$ according to Definition  \ref{deg:gs}. Then, we will say that the ground state is \emph{gapped} if
\begin{equation}\label{eq:gap_02}
    {\rm Spec}\big(\pi_\omega(H)\big)\;\cap\;(\lambda_\ast,\lambda_\ast+\Delta)\;=\;\emptyset\;.
\end{equation}
for some $\Delta>0$. Moreover, if
\[
{\rm Ker}(\pi_\omega(H)-\lambda_*{\bf 1})\;=\:\{\psi_\omega\}
\] 
we will say that the  gapped ground state is \emph{non-degenerate}.
\end{definition}

\medskip

Let us denotes with $\mathcal{K}_\omega:={\rm Ker}(\pi_\omega(H)-\lambda_*{\bf 1})^\bot$ the 
orthogonal complement in $\mathcal{H}_\omega$  to
the eigenspace associated with the ground state $\lambda_\ast$. By observing that the infimum of  ${\rm Spec}\big(\pi_\omega(H)\big)\setminus \{\lambda_\ast\}$ can be estimated with the help of \cite[Theorem 2.19]{Tes} applied to the subspace   $\mathcal{K}_\omega$, one can reformulate the gap condition by
\begin{equation}\label{eq:gpa_ineq}
\inf_{ \substack{\phi\in\mathcal{K}_\omega\\\|\phi\|_{\s{H}_\omega}=1}} \langle\phi,\pi_\omega(H)\phi\rangle_{\s{H}_\omega}\;\geqslant\;\lambda_\ast+\Delta\;.
\end{equation}

\medskip

The next result shows that also the gap condition can be characterized by a dynamical property \cite[Section 6]{Fannes-nachtergaele-werner-92}.

\begin{proposition}\label{Prop gs}
Let  $H=H^*$ be a self-adjoint element of the $C^*$-algebra $\rr{A}$ and $\omega\in \n{E}_{\rr{A}}$ a ground state of $H$. Then, $\omega$ is a 
non-degenerated {gapped ground state}  if and only there exists a $\Delta >0$ such that
\begin{equation}\label{eq:gap_01}
-\ii\;\omega\big(A^*\delta_H(A)\big)\;\geqslant\;\Delta\left(\omega\big(A^*A\big)-|\omega(A)|^2\right)\;, \qquad \forall\; A\in\rr{A}\;.  
\end{equation}
\end{proposition}
\begin{proof}
 Let us start with the implication $(\Rightarrow)$. Since the ground state is assumed to be  non-degenerated one has that $\mathcal{K}_\omega\subset \mathcal{H}_\omega$ be the orthogonal complement of the one-dimensional subspace generated by $\psi_\omega$. 
 Let $\psi_A:=\pi_\omega(A)\psi_\omega$ and 
 $\psi^\bot_A:=\psi_A-\langle \psi_\omega,\psi_A\rangle_{\s{H}_\omega} \psi_\omega$ its orthogonal projection  on $\mathcal{K}_\omega$. In view of \eqref{eq:gpa_ineq} one has that 
 \begin{equation}\label{eq:gapp_001}
\big\langle \psi_A^\bot,\big(\pi_\omega(H)-(\lambda_*+\Delta){\bf 1}\big)\psi_A^\bot\big\rangle_{\s{H}_\omega}\;\geqslant\; 0 
\end{equation}
 for every $A\in\rr{A}$. However, a direct computation shows that
 \begin{equation}\label{eq:gapp_002}
\big\langle \psi_A^\bot,\psi_A^\bot\big\rangle_{\s{H}_\omega}\;=\; \omega(A^*A)\;-\;|\omega(A)|^2 
 \end{equation}
 and
 \begin{equation}\label{eq:gapp_003}
 \begin{aligned}
 \big\langle \psi_A^\bot,\big(\pi_\omega(H)-\lambda_*{\bf 1}\big)\psi_A^\bot\big\rangle_{\s{H}_\omega}\;&=\; \big\langle \psi_A,\big(\pi_\omega(H)-\lambda_*{\bf 1}\big)\psi_A\big\rangle_{\s{H}_\omega}\\
&=\; \omega(A^*(HA-AH))
\end{aligned}
  \end{equation}
 where in the last equality we used that $\lambda_\ast\omega(A^*A)=\omega(A^*AH)$ in view of Proposition \ref{prop:gs}. By putting together inequality  \eqref{eq:gapp_001} with \eqref{eq:gapp_002} and \eqref{eq:gapp_003} one finally gets \eqref{eq:gap_01}.
For the implication  $(\Leftarrow)$ let us observe that in view of \eqref{eq:gapp_002} and \eqref{eq:gapp_003} one has that condition \eqref{eq:gap_01} is equivalent to condition \eqref{eq:gapp_001}.
Now,  in view of the cyclicity of $\psi_\omega$, for any $\phi\in \mathcal{K}_\omega$ there is a sequence $\{A_n\}_{\n\in\N}\subset \mathfrak{A}$ such that  $\psi_{A_n}\rightarrow\phi$. This immediately implies that $\psi_{A_n}^\bot\rightarrow\phi$. Therefore, one has that 
 \[
 \big \langle \phi,\big(\pi_\omega(H)-(\lambda_*+\Delta){\bf 1}\big)\phi\big\rangle_{\s{H}_\omega}\;=\;\lim_{n\rightarrow\infty}\big \langle \psi_{A_n}^\bot,\big(\pi_\omega(H)-(\lambda_*+\Delta){\bf 1}\big)\psi_{A_n}^\bot\big\rangle_{\s{H}_\omega}\;\geqslant\; 0 \]
 in view of \eqref{eq:gapp_001}. The positivity of 
 the quantity on the left-hand side for every  $\phi\in \mathcal{K}_\omega$ implies \eqref{eq:gpa_ineq} which is equivalent to the gap condition.
 \end{proof}

\begin{remark}\label{rem 22}
The condition of non-degeneracy of the ground state in Proposition \ref{Prop gs} is  necessary for the implication $(\Rightarrow)$. In fact, if ${\rm dim}\;{\rm Ker}(\pi_\omega(H)-\lambda_*{\bf 1})>1$ there is 
a non-zero $\psi'\in {\rm Ker}(\pi_\omega(H)-\lambda_*{\bf 1})$ which is orthogonal to $\psi_\omega$. For simplicity let us assume that there exists a $A\in\rr{A}$ such that $\psi':=\psi_A$. In this case, the same computation in  \eqref{eq:gapp_003} provides $\omega\big(A^*\delta_H(A)\big)=0$
and the orthogonality between $\psi_A$ and $\psi_\omega$ reads $\omega(A)=0$. In summary, one obtains
\[
|\omega(A)|^2\;-\;\ii\Delta^{-1}\omega\big(A^*\delta_H(A)\big)\;=\;0\;<\;\|\psi_A\|_{\s{H}_\omega}^2\;=\;\omega(A^*A)
\]
which contradicts \eqref{eq:gap_01}.
The general case just follows by a density argument similar to that used in the second part of the proof of Proposition \ref{Prop gs}.
\hfill $\blacktriangleleft$
\end{remark}
%

%-----------------------------------%
\section{Fermi surfaces and eigenstates}\label{secFS}
In this section, we introduce the notion of \emph{generalized (or noncommutative) Fermi surface} for an abstract $C^*$-algebra $\A$. We prove that under suitable conditions on $\A$, we can construct an eigenstate associated with a  Fermi surface. Classical references in the physical literature about the notion of Fermi surfaces are \cite{Ash,Cal,Kit}. From the mathematical point of view, we will refer to the review 
 \cite{Kuc}.

\subsection{Standard notion of Fermi surface}\label{sec:stanFer}
Let us start  by introducing the concept of Fermi surface in a fairly standard framework.
Let $X$ be a compact metric  space and  consider the $C^*$-algebra 
\begin{equation}\label{eq:nic-alg}
\rr{A}\;:=\;C(X)\otimes \rr{K}\;\simeq\;C(X,\rr{K})\;.
\end{equation}
The isomorphism above \cite[Theorem 6.4.17]{Murphy} allows us to think of $\rr{A}$ as the $C^*$-algebra of continuous functions on $X$ with values on the compact operators $\mathfrak{K}$ (on some sperable Hilbert space $\s{H}$). 
Since we will need to consider the spectrum of the elements in $\rr{A}$ and the states of $\rr{A}$, we will consider $\rr{A}$ as a subalgebra of its standard unitization $\rr{A}^+$ \cite[Proposition 2.1.5]{bratteli-robinson-87} made by continuous functions of the type $X\ni x\mapsto A(x)+\lambda{\bf 1}$ with  $A(x)\in \mathfrak{K}$ and $\lambda\in\C$.
These technicalities become irrelevant if one replaces $\mathfrak{K}$ with the matrix algebra ${\rm Mat}_n(\C)$.
For every $A\in \A$  the following spectral equality holds true
\begin{equation}\label{spec: eq}
   {\rm Spec}(A)\;=\;\bigcup_{x\in X}{\rm Spec}\big(A(x)\big)\;.
 \end{equation}
 Equation \eqref{spec: eq} can be deduced from a very general claim about the spectra of direct integrals of operators \cite[Theorem XIII.85]{reed-simon-IV}. However, it can also be justified directly.
 First of all, since $A\mapsto A(x)$ is a $*$-homomorphism for every fixed $x\in X$, one has that  ${\rm Spec}(A(x))\subset {\rm Spec}(A)$
 and in turn 
$\bigcup_{x\in X}{\rm Spec}(A(x))\subseteq {\rm Spec} (A)$. The reverse inclusion follows from \cite[Theorem 1.4.6]{Kuch2} which claims that if $\lambda \notin {\rm Spec}(A(x))$ for all $x\in X$, namely if  $A(x)-\lambda {\bf 1}$ is invertible for all $x\in X$, then $A-\lambda {\bf 1}$ is invertible, and in turn  $\lambda \notin {\rm Spec}(A)$. Therefore, one gets $\bigcup_{x\in X}{\rm Spec}(A(x))\supseteq {\rm Spec} (A)$.
Given the equality \eqref{spec: eq} we can introduce the  \emph{Fermi surface} (or \emph{variety}) of $A$ at the spectral point $\lambda\in {\rm Spec}(A)$ as the subset of $X$ defined by
\begin{equation}\label{eq:FS}
\bb{F}^A_{\lambda}\;:=\;\{x\in X\;|\; \lambda\in {\rm Spec}(A(x))\}\;\subseteq\; X\;.
\end{equation}

\medskip

The concept of  Fermi surface originates initially in the study of electrical conduction in metals and is a central concept in solid state physics \cite{Ash,Kit}. In fact, the Hamiltonians describing the dynamics of an electron in $d$-dimensional crystals are linear elliptic partial differential operators with $\Z^d$-periodic coefficients. The Bloch-Floquet transform maps (the resolvent of) such operators in elements of the algebra $\rr{A}=C(\n{T}^d,\rr{K})$ \cite{Kuc,Kuch2}. The $d$-dimensional torus $\n{T}^d$, which appears in this context as the Pontryagin dual of the symmetry group  $\Z^d$, 
 is known as the \emph{Brillouin zone}. Therefore, in condensed matter physics the 
Fermi surfaces are the \emph{isoenergy} level sets contained in the Brillouin zone of a given periodic Hamiltonian.
A toy model of this type is described in Section \ref{sec:graf}.

\subsection{Eigenstates for Fermi surfaces}\label{asec:class_eig}
Our next task is to show that under certain circumstances it is possible to associate an eigenstate to a Fermi surface defined by a self-adjoint operator.

\medskip

As in the previous section let us focus on the $C^*$-algebra $\A=C(X,\rr{K})$ with $X$ a 
compact metric space. Let $\bb{B}_X$ be the Borel $\sigma$-algebra of $X$ and  assume that $X$ is endowed with a probability Borel measure $\mu$.
Let $H=H^*$ be a self-adjoint element in $\A$, $\lambda\in{\rm Spec}(H)$ a spectral point and $\bb{F}^H_{\lambda}$ the associated Fermi surface according to \eqref{eq:FS}. 
Since $H(x)$ is a compact operator for every $x\in X$, 
one has that ${\rm Spec}(H(x))$ is made by an increasing  sequence of real eigenvalues
 \[
\varepsilon_{\rm min} \;\leq\; \varepsilon_0(x)\;\leq\; \varepsilon_1(x)\;\leq\;\ldots\;\leq\;\varepsilon_n(x)\;\leq\;\ldots\;\leq\;\varepsilon_{\rm max}\;,
\]
that are repeated according to their  multiplicity,
and $\varepsilon_{\rm min}$ and $\varepsilon_{\rm max}$ are the minimum and the maximum of ${\rm Spec}(H)$, respectively. We will refer to every $\varepsilon_n:X\to\R $ as the $n$-\emph{th energy function}. Standard 
perturbation theory provides that the {energy function} $\varepsilon_n$ are 
continuous \cite{kato-80,Kuch2}. Therefore $\varepsilon_n^{-1}(\{\lambda\})$
is a closed (hence Borelian) subset of $X$ for every $\lambda\in\R$, and evidently $\varepsilon_n^{-1}(\{\lambda\})=\emptyset$ whenever $\lambda\notin {\rm Spec}(H)$. In view of the equality \eqref{spec: eq} and the definition \eqref{eq:FS} one gets that
\[
\bb{F}^H_{\lambda}\;=\;\bigcup_{n=0}^{+\infty}\varepsilon_n^{-1}(\{\lambda\})\;,
\]
and as a consequence one concludes that every Fermi surface is an element of the  Borel $\sigma$-algebra $\bb{B}_X$.
\begin{assumption}[Finite band contribution]\label{ass:FBC}
There is only a finite number $M$ of energy functions $\{\varepsilon_{n_1},\ldots,\varepsilon_{n_M}\}$ such that $\varepsilon_{n_j}^{-1}(\{\lambda\})\neq \emptyset$.
\end{assumption}
\noindent
The consequence of the assumption above is that $\bb{F}^H_{\lambda}$ is a closed, hence  compact  subset of $X$.

\medskip

Consider now the function   $f_{H,\lambda}:X\to\R$ defined by
\[
f_{H,\lambda}(x)\;:=\;\prod_{j=1}^M\left(\varepsilon_{n_j}(x)-\lambda\right)\;.
\]
This is evidently a  continuous function (product of continuous functions), hence measurable,
such that $f_{H,\lambda}^{-1}(\{0\})=\bb{F}^H_{\lambda}$.
Theorem \ref{teo: desintegration} assures that there exists a disintegration $\{\mu_t^H\}_{t\in\R}$ of $\mu$ subordinated to 
$f_{H,\lambda}:X\to\R$ such that $\mu_t^H$ is supported in $f_{H,\lambda}^{-1}(\{t-\lambda\})$ (the shift in the variable is for notation purposes only). In particular, $\mu_\lambda^H$ turns out to be a probability measure supported on $\bb{F}^H_{\lambda}$. 
\begin{definition}[Fermi measure]\label{def:F-Mes}
The probability measure $\mu_\lambda^H$ on $X$ defined above will be called the Fermi measure associated with the Fermi surface $\bb{F}^H_{\lambda}$.
\end{definition}

\medskip

In order to go further we need to assume a gap condition subordinated to the Fermi surface.
\begin{assumption}[Local gap condition]\label{ass:loc_gap}
There is a positive constant $g>0$ such that
\[
\inf_{x\in \bb{F}^H_{\lambda}}\big\{{\rm dist}\left(\{\lambda\}, {\rm Spec}(H(x))\setminus \{\lambda\}\right)\big\}\;=\; g\;>\;0\;.
\]
Equivalently, on each point of the Fermi surface the eigenvalue $\lambda$ is separated by a finite gap of size at least $2g$ from the rest of the spectrum of $H(x)$.  
\end{assumption}
\noindent
Let $Q:\R\to\R$ be a continuous function such that
\[
Q(t)\;:=\;\left\{
\begin{aligned}
&1&\quad&\text{if}\; t\in \left[\lambda-\frac{g}{3},\lambda+\frac{g}{3}\right]\;,\\
&0&\quad&\text{if}\; t\in \left(-\infty,\lambda-\frac{2}{3}g\right]\cup\;\left[\lambda+\frac{2}{3}g,+\infty\right)\;.\\
\end{aligned}
\right.
\]
By functional calculus, one gets that $Q(H)\in\A$. Moreover, since the evaluation at  $x$ is a homomorphism of $C^*$-algebras, and homomorphisms commute with the functional calculus, one obtains that 
$Q(H)(x)=Q(H(x))$. Taking into account Assumption \ref{ass:loc_gap} one infers that 
\[
Q(H)(x)\;=\;P_\lambda^H(x)\;,\qquad x\in \bb{F}^H_{\lambda}\;,
\]
where $P_\lambda^H(x)$ denotes the spectral projection of $H(x)$ for the eigenvalue $\lambda$. In particular, in view of the compactness of $H(x)$ the projection $P_\lambda^H(x)$ is finite-dimensional, hence trace class, whenever $\lambda\neq 0$. 

\medskip

Let us introduce the density matrix
\begin{equation}\label{eq:rho3}
\rho_\lambda^H(x)\;:=\;\left\{
\begin{aligned}&\frac{P_\lambda^H(x)}{{\rm Tr}(P_\lambda^H(x))}&\quad&\text{if}\; x\in \bb{F}^H_{\lambda}\;,\\
&0&\quad&\text{otherwise}\;.\\
\end{aligned}
\right.
\end{equation}
\begin{lemma}\label{lemma:meas}
Let Assumptions \ref{ass:FBC} and \ref{ass:loc_gap} be valid. Then, for every $A\in\A$ the map
\begin{equation}\label{eq:main_map}
X\;\ni\;x\;\longmapsto\; {\rm Tr}\left(\rho_\lambda^H(x) A(x)\right)\;\in\;\C
\end{equation}
is a measurable function.
\end{lemma}
\proof
Since by Assumption $\ref{ass:FBC}$ the Fermi surface $\bb{F}^H_{\lambda}\subset X$ is compact, there is a restriction map $\iota:C(X,\rr{K})\to C(\bb{F}^H_{\lambda},\rr{K})$ which is a homomorphism of $C^*$-algebras. Said differently $\iota(Q)=Q|_{\bb{F}^H_{\lambda}}$ is a continuous map and 
$\iota(Q(H))(x)=P_\lambda^H(x)$ is trace class. Therefore the map
\[
\bb{F}^H_{\lambda}\;\ni\;x\;\longmapsto\; {\rm Tr}\left(\iota(Q(H)A)(x)\right)\;=\; {\rm Tr}\left(P_\lambda^H(x)\iota(A)(x)\right)\;\in\;\C
\]
is continuous, hence measurablel on $\bb{F}^H_{\lambda}$. The same is true for the map
\[
\bb{F}^H_{\lambda}\;\ni\;x\;\longmapsto\; {\rm Tr}\left(\iota(Q(H))(x)\right)\;=\; {\rm Tr}\left(P_\lambda^H(x)\right)\;\neq\;0\;.
\]
As a result 
\begin{equation}\label{eq:main_map2}
\bb{F}^H_{\lambda}\;\ni\;x\;\longmapsto\; \frac{{\rm Tr}\left(P_\lambda^H(x)\iota(A)(x)\right)}{{\rm Tr}\left(P_\lambda^H(x)\right)}\;\in\;\C\;
\end{equation}
is a measurable function.
The proof is concluded by observing that the map \eqref{eq:main_map} is obtained by multiplying the map \eqref{eq:main_map2} by the characteristic function $\chi_{\bb{F}^H_{\lambda}}$ of the Fermi surface $\bb{F}^H_{\lambda}$, the latter being  a measurable function on $X$.\qed

\medskip
We are now in a position to prove that Fermi surfaces are related to the existence of specific eigenstates.

\begin{theorem}\label{prop: eigenstate fermi_S}
Assume that $X$ is a compact  metric space endowed with a 
Borel probability measure $\mu$. Let $H=H^*$ be a self-adjoint element of the $C^*$-algebra $\A:=C(X,\rr{K})$,  $\lambda\in {\rm Spec}(H)\setminus\{0\}$ a spectral point, $\bb{F}^H_{\lambda}$ the associated Fermi surface defined by \eqref{eq:FS} and $\mu_\lambda^H$ the Fermi measure described in Definition \ref{def:F-Mes}. If Assumptions \ref{ass:FBC} and \ref{ass:loc_gap} are satisfied then the functional
\begin{equation}
    \omega_\lambda(A)\;:=\;\int_{X}\dd \mu^H_\lambda(x)\;{\rm Tr}\left(\rho_\lambda^H(x) A(x)\right)\;,\qquad A\in \A\;,
\end{equation}
with $\rho_\lambda^H$ defined by \eqref{eq:rho3},
provides an eigenstate of $H$ with eigenvalue $\lambda$.
\end{theorem}
\proof
First of all  the map $ \omega_\lambda:\A\to\C$ in \eqref{eigen_fer_S} is well defined in view of Lemma \ref{lemma:meas} and the fact that the 
Fermi measure is supported exactly on $\bb{F}^H_{\lambda}$.
 It is evidently linear and normalized
since
\[
 \omega_\lambda({\bf 1})\;=\;\int_{X}\dd \mu^H_\lambda(x)\;=\;\mu^H_\lambda(X)\;=\;1.
\]
In the last equation  we have tacitly identified $ \omega_\lambda$ with its canonical extension on $\A^+$ \cite[Corollary 2.3.13]{bratteli-robinson-87}.
 It is also positive since ${\rm Tr}(\rho_\lambda^H(x) A(x))\geqslant 0$ for every $x\in X$ if $A$ is a positive element. Summing up, it turns out that  $ \omega_\lambda$ is a state of $\A^+$. It remains to prove that $ \omega_\lambda$ is an eigenstate of $H$. Observe that when $x\in\bb{F}^H_{\lambda}$ one has that $P_\lambda^H(x)$ is the spectral projection of $H(x)$ at the eigenvalue $\lambda$ and therefore
 \[
 {\rm Tr}\left(\rho_\lambda^H(x) A(x)H(x)\right)\;=\; {\rm Tr}\left(H(x)\rho_\lambda^H(x) A(x)\right)\:=\;\lambda  {\rm Tr}\left(\rho_\lambda^H(x) A(x)\right)\;.
 \]
 On the other hand, the Fermi measure is supported exactly on $\bb{F}^H_{\lambda}$. As a result, one has that
 \[
 \begin{aligned}\label{eigen_fer_S}
    \omega_\lambda(AH)\;&=\;\int_{X}\dd \mu^H_\lambda\;{\rm Tr}\left(\rho_\lambda^H(x) A(x)H(x)\right)\;=\;\int_{\bb{F}^H_{\lambda}}\dd \mu^H_\lambda\;{\rm Tr}\left(\rho_\lambda^H(x) A(x)H(x)\right)\\
    &=\;\lambda\int_{\bb{F}^H_{\lambda}}\dd \mu^H_\lambda\;{\rm Tr}\left(\rho_\lambda^H(x) A(x)\right)\;=\;\lambda\int_{X}\dd \mu^H_\lambda\;{\rm Tr}\left(\rho_\lambda^H(x) A(x)\right)\;=\;\lambda\omega_\lambda(A)
\end{aligned}
 \]
 for every $A\in\R$. This concludes the proof.
\qed

\begin{definition}[Fermi eigenstate]
The eigenstate $\omega_\lambda$ defined by \eqref{eigen_fer_S} will be called the \emph{Fermi eigenstate} of $H$ at the spectral value $\lambda\in {\rm Spec}(H)\setminus\{0\}$.
\end{definition}

\begin{remark}
The crucial point of the proof of Theorem \ref{prop: eigenstate fermi_S} is to prove that the map
\[
\bb{F}^H_{\lambda}\;\ni\;x\;\longmapsto\;  {\rm Tr}\left(\rho_\lambda^H(x) A(x)\right)\;\in\;\C
\]
is measurable. Assumptions \ref{ass:FBC} and \ref{ass:loc_gap} are sufficient to ensure this fact but are not necessary in principle. \hfill $\blacktriangleleft$
\end{remark}

%--------%
\subsection{A physical application: the graphene}\label{sec:graf}
Let $\rr{s}_1$ and $\rr{s}_2$ be the shift operators defined on $\ell^2(\Z^2)$ by
$(\rr{s}_j\psi)(n):=\psi(n-e_j)$ with $j=1,2$, $\psi\in \ell^2(\Z^2)$ and $e_1:=(1,0)$ and 
$e_2:=(0,1)$ the canonical basis of $\Z^2$. We will denote by $C^*(\rr{s}_1,\rr{s}_2)$ the unital $C^*$-algebra generated by the shift operators. Elements of $C^*(\rr{s}_1,\rr{s}_2)$ are usually called \emph{periodic} operators. In order to provide a mathematical model for the graphene we need to extend the algebra of  {periodic} operators to allow for the isospin degree of freedom. 
As explained in detail in \cite{denittis-lein-13} this is accomplished by considering the algebra $\rr{A}:=C^*(\rr{s}_1,\rr{s}_2)\otimes {\rm Mat}_2(\C)$. The \emph{isotropic nearest-neighbor model}  
perturbed by a \emph{stagger potential} $\gamma>0$ is described by the self-adjoint operator
\begin{equation}\label{op:graf}
H\;:=\;\begin{pmatrix}
+\gamma{\bf 1}& \rr{s}_1+\rr{s}_2&\\
\rr{s}_1^*+\rr{s}_2^*&-\gamma{\bf 1}
\end{pmatrix}\;.
\end{equation}
The Hamiltonian \eqref{op:graf}  can also be understood  as a two dimensional version of the Su-Schrieffer-Heeger (SSH) model \cite{su-schrieffer-heeger-79}.

\medskip

Consider the (inverse) Fourier transform $\bb{F}:\ell^2(\Z^2)\to L^2(\n{T}^2)$ where 
$\n{T}^2\simeq[0,2\pi)^2$ is the two-dimensional torus, parameterized by the coordinates $k=(k_1,k_2)$, and endowed with the normalized Haar measure
$\dd\mu(k):=(2\pi)^{-2}\dd k$. Under this transform one has that the shift operators are mapped to the multiplication by the phases $\expo{-\ii k_1}$ and $\expo{-\ii k_2}$, \ie
$\bb{F}\rr{s}_j\bb{F}^*={\expo{-\ii k_j}}$ for $j=1,2$. As a consequence one gets the isomorphism $\A\simeq C(\n{T}^2)\otimes {\rm Mat}_2(\C)$ and the operator \eqref{op:graf}
is mapped to the function $k\mapsto H(k)$ with
\begin{equation}\label{op:graf2}
H(k)\;:=\;\begin{pmatrix}
+\gamma{\bf 1}& \expo{-\ii k_1}+\expo{-\ii k_2}&\\
\expo{+\ii k_1}+\expo{+\ii k_2}&-\gamma{\bf 1}
\end{pmatrix}\;.
\end{equation}

\medskip

Since we are in the precise scenario of Section \ref{asec:class_eig} we can compute the Fermi surfaces and the related measures for the operator \eqref{op:graf2}. First of all it is easy to verify that $H(k)$ as two energy bands $\varepsilon_-(k)<\varepsilon_+(k)$ given by
\begin{equation}\label{eq:en_band}
\varepsilon_\pm(k)\;:=\;\pm\sqrt{\gamma^2+4\cos^2\left(\frac{k_1-k_2}{2}\right)}
\;
\end{equation}
and the related spectral projections are
\begin{equation}\label{eq:eigen_band}
P_\pm(k)\;:=\;\frac{1}{2}{\bf 1}_2\;+\; \frac{1}{2 \varepsilon_\pm(k)}H(k)
\;.
\end{equation}
The energy band $\varepsilon_\pm$ takes al the values between the limit values $\pm\gamma$ and $\pm\sqrt{\gamma^2+4}$. Therefore the spectrum of $H$ is given by
\[
\sigma(H)\;:=\;\left[-\sqrt{\gamma^2+4},-\gamma\right]\;\cup\;\left[\gamma,\sqrt{\gamma^2+4}\right]
\]
and shows a central gap of size $2\gamma$.

\medskip

Let us construct the Fermi surfaces for the upper band $\varepsilon_+$. The construction for the lower band $\varepsilon_-$ is totally similar.
For every $\lambda\in Y:=[\gamma,\sqrt{\gamma^2+4}]$, the Fermi surface 
$\bb{F}^H_{\lambda}$ is the null-set of the function $f_{H,\lambda}(k):=\varepsilon_+(k)-\lambda$, namely it is described by the solutions of the equations
\[
\cos\left(\frac{k_1-k_2}{2}\right)\;=\;\pm\frac{\sqrt{\lambda^2-\gamma^2}}{2}\;.
\]
Let us fix $\theta\equiv \theta(\lambda)$ as
\[
\theta\;:=\;\arccos\left(\frac{\sqrt{\lambda^2-\gamma^2}}{2}\right)\;\in\;\left[0,\frac{\pi}{2}\right]\;
\]
and consider the linear subsets
\[
\begin{aligned}
L_\pm\;&:=\;\{k\in\n{T}^2\;|\;k_2=k_1\pm2\theta\}\;,\\
G_\pm\;&:=\;\{k\in\n{T}^2\;|\;k_2=k_1\pm2(\pi-\theta)\}\;
\end{aligned}
\]
that depend on $\lambda$ through $\theta$.
Then, one gets that  
\[
\bb{F}^H_{\lambda}\;=\; L_+\;\cup\;L_-\;\cup\;G_+\;\cup\;G_-\;.
\]
The two extreme cases are given by $\lambda_{\rm min}:=\gamma$ which corresponds to $\theta=\frac{\pi}{2}$,  and $\lambda_{\rm max}:=\sqrt{\gamma^2+4}$ which provides $\theta=0$. In the first case, one obtains that the lines $L_\pm$ coincide with the lines $G_\pm$
and as a result one gets $\bb{F}^H_{\lambda_{\rm min}}=L_{+}^{\rm min}\cup L_{-}^{\rm min}$ where
\[
L_{\pm}^{\rm min}\;:=\;\{k\in\n{T}^2\;|\;k_2=k_1\pm \pi\}\;.
\]
In the second case, the four lines collapse into the single line
\[
L^{\rm max}\;:=\;\{k\in\n{T}^2\;|\;k_2=k_1\}\;
\]
which coincides with $\bb{F}^H_{\lambda_{\rm max}}$.

\medskip

The next task is to compute the Fermi measures induced by the upper  energy band.  For simplicity let us exclude the 
extreme cases by assuming $\lambda\neq \lambda_{\rm min},\lambda_{\rm max}$.
First, let us compute the expectation operator $E$ induced by $\varepsilon_+$. Observe that  $\varepsilon_+=f\circ g$, where  $g(k_1,k_2)=k_1-k_2$ and $f(y)=\sqrt{\gamma^2+4\cos^2\big(\frac{y}{2}\big)}$. Therefore, by invoking  Proposition \ref{prop: composition}, one gets $E=E^fE^g$. The expectation operators $E^g$ and $E^f$ are computed in  Examples \ref{ex: a} and \ref{dis: function}, respectively. By using these results one obtains  that 
\begin{equation}
\begin{split}
       E_\psi(\lambda)\;&=\;E^f_{E^g_\psi}(\lambda)\;=\;\frac{1}{4}\sum_{y\in P_\lambda}E^g_\psi(y)
\end{split}
\end{equation}
with $\psi\in L^1(\T^2,\mu)$,  $\lambda\in Y$
 a given spectral value and
 \[
 P_\lambda\;:=\;f^{-1}(\{\lambda\})\;=\;\big\{\pm2\theta,\pm2(\pi-\theta)\big\}\;.
 \]
On the other hand, from the computation in  Example \ref{ex: a}, and some manipulation, one obtains
\begin{equation*}
    \begin{split}
   E^g_\psi(\pm2\theta) \;&=\;     \int_{0}^{1}\dd \tau\;\psi\big(s_{L_\pm}(\tau)\big)\\
    \end{split}
\end{equation*}
where  $[0,1]\ni \tau\mapsto s_{L_\pm}(\tau)\in\n{T}^2$ is a linear parametrization of $L_\pm$
More precisely one has that
\[
\begin{aligned}
s_{L_+}(\tau)\;&:=\;\big(\tau(2\pi-2\theta), \tau(2\pi-2\theta)+2\theta \big)\;,\\
s_{L_-}(\tau)\;&:=\;\big(\tau(2\pi-2\theta)+2\theta, \tau(2\pi-2\theta) \big)\;.
\end{aligned}
\]
Similarly one has that
\begin{equation*}
    \begin{split}
   E^g_\psi\big(\pm2(\pi-\theta)\big) \;&=\;     \int_{0}^{1}\dd \tau\;\psi\big(s_{G_\pm}(\tau)\big)\\
\end{split}
\end{equation*}
with
\[
\begin{aligned}
s_{G_+}(\tau)\;&:=\;\big(\tau 2\theta, \tau 2\theta+(2\pi-2\theta) \big)\;,\\
s_{G_-}(\tau)\;&:=\;\big(\tau 2\theta+(2\pi-2\theta), \tau2\theta \big)\;.
\end{aligned}
\]
Putting the ingredients all together one obtains
\begin{equation}
E_\psi(\lambda)\;=\;\frac{1}{4}\sum_{\sharp=\pm}\int_{0}^{1}\dd \tau\;\left[\psi\big(s_{L_\sharp}(\tau)\big)+\psi\big(s_{G_\sharp}(\tau)\big)\right]\;
\end{equation}
and the disintegration $\mu^H_\lambda$ of the Haar measure 
induced by the energy band $\varepsilon_+$
takes the form
\[
\mu^H_\lambda(k_1,k_2)\;=\;\frac{1}{4}\sum_{\sharp=\pm}\int_{0}^{1}\dd \tau\;\left[\delta_{s_{L_\sharp}(\tau)}(k_1,k_2)\;+\;\delta_{s_{G_\sharp}(\tau)}(k_1,k_2)\right]
\]
where $\delta_{k}$ denotes the Dirac measure concentrated at the point $k\in\n{T}^2$.
Evidently, $\mu^H_\lambda$ is concentrated on the Fermi surface $\bb{F}^H_{\lambda}$.

\medskip

With these ingredients,  the associated Fermi eigenstate $\omega_\lambda$  is given by
\begin{equation*}
\begin{split}
\omega_\lambda(A)\;:=&\;\int_{\bb{F}^H_{\lambda}}{\rm d}\mu_\lambda^H(k_1,k_2)\,{\rm Tr}_{\C^2}\big(P_+(k_1,k_2)A(k_1,k_2)\big) \\
=&\; \frac{1}{4}\sum_{\sharp=\pm}\int_{0}^{1}\dd \tau\,{\rm Tr}_{\C^2}\big(P_+A\big(s_{L_\sharp}(\tau)\big)+P_+A\big(s_{G_\sharp}(\tau)\big)\big)\\
\end{split}
\end{equation*}
 for any $A\in \A$.

%--------% 

\subsection{Generalized notion of Fermi surface}\label{sec:GenFer}
We are now in a position to provide a generalized, algebraic definition of Fermi surface. 

\medskip

For a $C^*$-algebra $\A$ let us denote by $\hat{\A}$ to the spectrum of $\A$, \ie, the set of unitary  equivalence  classes  of  non-zero irreducible  $*$-representations of $\A$ \cite[Section 5.4]{Murphy}. There is a privileged topology, called \emph{hull-kernel} (or \emph{Jacobson}) \emph{topology} which makes $\hat{\A}$ a  topological space. Let ${\rm Prim}(\A)$ be the set of {primitive ideals} of $\A$. Then $\rr{I}\in {\rm Prim}(\A)$ if and only if there is a non-zero irreducible representation $\pi$ of $\A$ on the Hilbert space $\s{H}_\pi$ such that 
$\rr{I}={\rm Ker}(\pi)$ \cite[Theorem 5.4.2]{Murphy}. The latter characterization will be used here as the \virg{handy definition} of {primitive ideal}. Notice that primitive ideals are automatically closed and proper.
Since any two unitary equivalent representations have the same kernel, it makes sense to use the notation ${\rm Ker}([\pi])$ where $[\pi]$ is the equivalence class of $\pi$. This fact provides a  canonical surjection $\theta:\hat{\A}\to {\rm Prim}(\A)$ defined by $\theta:[\pi]\mapsto {\rm Ker}([\pi])$. 
The set ${\rm Prim}(\A)$ can be endowed with a unique topology that satisfies the following property: for each subset $\bb{R}\subseteq {\rm Prim}(\A)$ the closure of 
$\bb{R}$ is given by the set of primitive ideals of $\A$ which contain the intersection of the ideals in $\bb{R}$ \cite[Theorem 5.4.6]{Murphy}, that is
\begin{equation}\label{eq:clos}
\overline{\bb{R}}\;:=\;\left\{\rr{P}\in{\rm Prim}(\A)\;\left|\; \bigcap_{\rr{I}\in\bb{R}}\rr{I}\;\subseteq\;\rr{P}\right\}\right.\;.
\end{equation}
 The (induced) {hull-kernel} topology on $\hat{\A}$ is the 
weakest topology making the canonical surjection $\theta$ continuous. The topological space 
$\hat{\A}$  can be quite bizarre. However, one has that 
$\hat{\A}$ is a $T_0$ space if and only if the canonical surjection $\theta$ is a homeomorphism, namely if and only if $\hat{\A}\simeq {\rm Prim}(\A)$ as  topological spaces  \cite[Theorem 5.4.9]{Murphy}. In particular, this is the case when $\hat{\A}$ is a Hausdorff space. Let us also mention that when 
$\A$ is a unital $C^*$-algebra then $\hat{\A}$ turns out to be compact (this follows from \cite[Theorem 5.4.8]{Murphy}).

\medskip

For every  $A\in\A$ and $[\pi]\in\hat{\A}$ let us denote with $A([\pi])$ the image of $A$ in the quotient $\A_{[\pi]}:=\A/ Ker([\pi])$. Given any representative of $\pi$ of the class $[\pi]$, it is immediate to check that $\pi$  
provides an isomorphism of $C^*$-algebras between  $\A_{[\pi]}$ and $\pi(\A)$. Therefore, for any $A\in\A$ one has that
\[
{\rm Spec}(A([\pi]))\;=\; {\rm Spec}(\pi(A))\;\subseteq\; {\rm Spec}(A)
\] independently of the choice of the representative $\pi$.
The following result provides a generalization of \eqref{spec: eq}.

\begin{lemma}
For every $A\in\A$ it holds true that
\[
{\rm Spec}(A)\;=\;\bigcup_{[\pi]\in \hat{\A}}{\rm Spec}(A([\pi]))\;.
\]
\end{lemma}
\proof
Let $\n{P}_\A$ be the set of pure states of $\A$, and for a given $\omega\in \n{P}_\A$  consider the associated irreducible GNS representation $(\s{H}_\omega,\pi_\omega,\psi_\omega)$. As in the proof of \cite[Theorem 2.1.10]{bratteli-robinson-87} let us consider the direct sum representation $(\s{H},\pi)$ given by
\[
\s{H}\;:=\;\bigoplus_{\omega\in \n{P}_\A}\s{H}_\omega\;,\qquad \rho\;:=\;\bigoplus_{\omega\in \n{P}_\A}\pi_\omega\;.
\]
By repeating verbatim the proof of \cite[Theorem 2.1.10]{bratteli-robinson-87}, and observing that the state $\omega_A$ in this proof can be chosen pure in view of  \cite[Lemma 2.1.23]{bratteli-robinson-87},
one gets that $\rho$ is a $\ast$-isomorphism. It follows that
\[
{\rm Spec}(A)\;=\;{\rm Spec}(\rho(A))\;=\;\bigcup_{\omega\in \n{P}_\A}{\rm Spec}(\pi_\omega(A))\;.
\]
Since the representations $(\s{H}_\omega,\pi_\omega,\psi_\omega)$ are irreducible one gets that ${\rm Spec}(A)$ is contained in the union of ${\rm Spec}(\pi(A))$ when $\pi$ runs over all the irreducible representations. Since the spectrum only depend on the equivalence class of the representation one obtains 
\[
{\rm Spec}(A)\;\subseteq\;\bigcup_{[\pi]\in \hat{\A}}{\rm Spec}(A([\pi]))\;.
\]
The other inclusion is trivial in view of the fact that ${\rm Spec}(\pi(A))\subseteq {\rm Spec}(A)$ for every (not necessarily irreducible) representation $\pi$.
\qed

\medskip

\noindent
The proof above can be shortened using \cite[Theorem A.38]{Rae}.

\medskip 

We are now in a position to generalize the concept of Fermi surfaces.

\begin{definition}[Generalized Fermi surface]\label{def:genFS}
Let $A\in\A$ and $\lambda\in {\rm Spec}(A)$. The Fermi surface $\bb{F}^A_{\lambda}$ defined by $A$ at the spectral point $\lambda$ is the subset of $\hat{\A}$
defined by
\begin{equation}\label{eq:GenFS}
   \bb{F}^A_{\lambda}\;:=\;\left.\left\{[\pi]\in \hat{\A}\;\right|\; \lambda\in {\rm Spec}\big(A([\pi])\big)\right\}\;\subseteq\;\hat{\A}\;.
\end{equation}
\end{definition}

\begin{example}[Relation with the standard case]
To justify Definition \ref{def:genFS} it is needed to show that 
\eqref{eq:GenFS} is equivalent to \eqref{eq:FS} when $\rr{A}$ is a $C^*$-algebra of the form \eqref{eq:nic-alg}. Let us start from the simpler case of the abelian $C^*$ algebra $C(X)$ with $X$ a compact Hausdorff space. Then it is known that
 $\widehat{C(X)}=\{\epsilon_x\;|\; x\in X\}$ where $\epsilon_x:\rr{A}\to\C$ is the \emph{evaluation} at $x$ defined by $\epsilon_x(f):=f(x)$ for every 
$f\in C(X)$ \cite[Example A.16]{Rae}. Therefore the map $x\mapsto \epsilon_x$ provides a bijection between $X$ and $\widehat{C(X)}$ which turns out to be a homeomorphism of topological spaces \cite[p. 213-214]{Rae}. This allows us to write  
$\widehat{C(X)}\simeq X$ and this is an incarnation of the Gelfand isomorphism
\cite[Theorem 2.1.11A]{bratteli-robinson-87}.
The more general case in which  $\rr{A}$ is of the form $C(X,{\rm Mat}_n(\C))$ or $C(X,\mathfrak{K})$
is discussed in  \cite[Examples A.23 \& A.24]{Rae}.
Also in this case one has that the irreducible representations of $\rr{A}$ are evaluations of the type $\epsilon_x(A)=A(x)$ for every $x\in X$,
up to unitary equivalences. Therefore, there is a homeomorphism  $X\simeq \widehat{\A}$ provided by 
 $x\mapsto [\epsilon_x]$. Putting all this information together one can see that \eqref{eq:FS}  can be indeed rewritten in the generalized form \eqref{eq:GenFS}.   
\hfill $\blacktriangleleft$
\end{example}

\begin{example}[Compact operators]
Let us consider the compact operators $ \mathfrak{K}\subset\rr{B}(\s{H})$ (for some Hilbert space $\s{H}$) and its standard unitization
 $\mathfrak{K}^+=\C{\bf 1}+\mathfrak{K}$. The argument of 
 \cite[Example A.15]{Rae} shows that 
 the identity representation ${\rm id}:\mathfrak{K}\to\rr{B}(\s{H})$ is irreducible and every irreducible representation is unitarily equivalent to ${\rm id}$. Consequently
 $\widehat{\mathfrak{K}}=\{[{\rm id}]\}$ and ${\rm Prim}(\mathfrak{K})=\{\rr{I}_0\}$ are (homeomorphic) singletons with $\rr{I}_0:=\{0\}$ 
 the trivial ideal.
 To compute ${\rm Prim}(\mathfrak{K}^+)$ let us observe that $\mathfrak{K}$ is an ideal in $\mathfrak{K}^+$. Then by virtue of \cite[Proposition A.26 (a)]{Rae} one gets
 \[
 \begin{aligned}
 {\rm Prim}(\mathfrak{K}^+)\;&=\;\{\rr{P}\in{\rm Prim}(\mathfrak{K}^+)\;\left|\; \mathfrak{K}\;\subseteq\;\rr{P} \}\right.\;\cup\;\{\rr{P}\in{\rm Prim}(\mathfrak{K}^+)\;\left|\; \mathfrak{K}\;\not\subset\;\rr{P} \}\right.\\
 &=\;\{\rr{P}\in{\rm Prim}(\mathfrak{K}^+)\;\left|\; \mathfrak{K}\;\subseteq\;\rr{P} \}\right.\;\cup\;{\rm Prim}(\mathfrak{K})\\
 &=\;\{\rr{K},\rr{I}_0\}\:.
 \end{aligned}
 \]
The irreducible representation with kernel $\rr{K}$ is the
  one-dimensional representation $\pi_0(K+\alpha{\bf 1})=\alpha$ for every $K\in\rr{K}$ and $\alpha\in\C$. 
  By  using \eqref{eq:clos}
one gets that
\[
\overline{\rr{I}_0}\;=\; \left\{\rr{P}\in{\rm Prim}(\mathfrak{K}^+)\;\left|\; \rr{I}_0\;\subseteq\;\rr{P}\right\}\right.\;=\;{\rm Prim}(\mathfrak{K}^+)
\]
which shows that $\rr{I}_0$ is dense in ${\rm Prim}(\mathfrak{K}^+)$.
 Therefore ${\rm Prim}(\mathfrak{K}^+)$  is not Hausdorff.
  However, it is $T_0$ and this provides the homeomorphism ${\rm Prim}(\mathfrak{K}^+)\simeq\widehat{\rr{K}^+}=\{[{\rm id}],[\pi_0]\}$. It is interesting to notice that 
 $\widehat{\rr{K}^+}$ provides an easy example of a non Hausdorff spectrum.
 Now let $A\in \rr{K}^+$ and $\lambda\in {\rm Spec}(A)$. Then, by applying the definition \eqref{eq:GenFS} one obtains that the Fermi surface $\bb{F}^A_{\lambda}$ defined by $A$ at the spectral point $\lambda$ is given by
\[
\bb{F}^A_{\lambda}\;=\;
\left\{
\begin{aligned}
&\{[{\rm id}]\}&\quad&\text{if}\;\pi_0(A)\neq\lambda\;,\\
&\widehat{\rr{K}^+}&\quad&\text{if}\;\pi_0(A)=\lambda\;.\\
\end{aligned}
\right.
\]
The presence of the dense point $[{\rm id}]$ independently of $\lambda$
follows from the fact that it corresponds to the identity representation  which preserves the spectrum.
\hfill $\blacktriangleleft$
\end{example}

\begin{example}[The Toeplitz $C^*$-algebra]
Let $\rr{T}$ be the \emph{Toeplitz $C^*$-algebra}, namely the $C^*$-subalgebra of $\rr{B}(\ell^2(\mathbb{N}))$ generated by the unilateral shift operator $\mathfrak{s}$  described in Example \ref{ex:uni-shif}.
There is a well-known  (not-split) exact sequence, called Toeplitz extension, given by
\begin{equation}
  0\;\longrightarrow\;\mathfrak{K} \;\stackrel{\iota}{\longrightarrow}\; \rr{T} \;\stackrel{\varphi}{\longrightarrow}\; C(\mathbb{S}^1) \;\longrightarrow\;0
\end{equation}
where $\mathfrak{K}$ is the algebra of compact operators on $\ell^2(\mathbb{N})$, $\iota$ is the inclusion (indeed $\mathfrak{K}\subset \rr{T}$) and
$\varphi$ is called the \emph{symbol map}  \cite[Theorem V.1.5]{davidson-96} and \cite[Proposition A.26]{Rae}.
As a consequence one has that $\rr{T}/\mathfrak{K}\simeq C(\mathbb{S}^1)$.  
Moreover, $\rr{T}$ acts irreducibly on $\ell^2(\mathbb{N})$ (it is a primitive $C^*$-algebra) meaning that $\rr{I}_0\in{\rm Prim}(\rr{T})$.
As in the example above one has that the trivial ideal is dense, \ie $\overline{\rr{I}_0}\;=\; {\rm Prim}(\rr{T})$.
Since $\rr{T}$ contains 
$\mathfrak{K}$ as its unique minimal (non-trivial) ideal, it follows that any other primitive ideal of $\rr{T}$ must contain $\mathfrak{K}$. 
This means that the associate irreducible representation must vanish on  $\mathfrak{K}$.
Therefore the other irreducible representations of $\rr{T}$ can be obtained by the combination of $\varphi$ with the evaluation at the points of $\mathbb{S}^1$.
More precisely one gets that $\pi_\theta:=\epsilon_\theta\circ\varphi:\rr{T}\to\C$ are irreducible representations of $\rr{T}$ for every $\theta\in \mathbb{S}^1$. Consequently one obtains that 
\[
{\rm Prim}(\rr{T})\;=\; \{\rr{I}_0\}\;\cup\;\{\rr{I}_\theta\;|\;\theta\in \mathbb{S}^1\}
\]
where $\rr{I}_\theta=\mathfrak{K}+\rr{L}_\theta$ and $\rr{L}_\theta$ satisfies
$\varphi(\rr{L}_\theta)=\{f\in C(\mathbb{S}^1)\;|\; f(\theta)=0 \}$.
One can check that ${\rm Prim}(\rr{T})\;\simeq\; \{\rr{I}_0\}\cup \mathbb{S}^1$
with the usual topology on $\mathbb{S}^1$ but with the property that the singleton $\{\rr{I}_0\}$  is dense. Therefore ${\rm Prim}(\rr{T})$  is not Hausdorff. However, it is $T_0$ and this provides the homeomorphism $\hat{\rr{T}}\simeq {\rm Prim}(\rr{T})$. Now let $A\in \rr{T}$ and $\lambda\in {\rm Spec}(A)$. Then  the Fermi surface $\bb{F}^A_{\lambda}$ defined by $A$ at the spectral point $\lambda$ is given by
\[
\bb{F}^A_{\lambda}\;=\;\{\rr{I}_0\}\;\cup\;\left.\left\{\theta\in\mathbb{S}^1 \;\right|\; \varphi(A)(\theta)=\lambda\right\}\;
\]
in view of the definition   \eqref{eq:GenFS}.
\hfill $\blacktriangleleft$
\end{example}

\medskip

The generalization of the notion of Fermi surface as given in Definition \ref{eq:GenFS} paves the possibility of extending the construction of the Fermi eigenstates in Theorem \ref{prop: eigenstate fermi_S} to more general $C^*$-algebras. This problem remains outside the scope of this work and provides material for new future investigations. However, it is worth anticipating that there is at least a class of $C^*$-algebras for which the extension of Theorem \ref{prop: eigenstate fermi_S} seems to be directly accessible, even if not entirely trivial. Following  \cite[Proposition 5.15]{Rae}
let us introduce the following definition:
\begin{definition}[Continuous-trace $C^*$-algebra]\label{def:CTCA}
A $C^*$-algebra $\A$ is called \emph{continuous trace} if and only if: (i) it spectrum $\hat\A$ is a Hausdorff space; (ii) $\A$ is locally Morita equivalent to $C_0(\A)$.
\end{definition}

\noindent
Let us clarify the meaning of condition (ii) in the definition above.
First of all notice that
 the open sets $\bb{U}\subset {\rm Prim}(\A)$ are in one-to-one correspondence with the ideals 
 \[
 \A_\bb{U}\;:=\;\bigcap\{ \rr{P} \in {\rm Prim}(\A)\;|\;\rr{P} \notin \bb{U}\}\;.
 \]
In other words, the topology on ${\rm Prim}(\A)$ always determines the ideal structure of $\A$. There are natural homeomorphisms $\bb{U}\simeq {\rm Prim}( \A_\bb{U})$ and ${\rm Prim}(\A)\setminus \bb{U}\simeq{\rm Prim}(\A/ \A_\bb{U})$ \cite[Proposition A.27]{Rae}.
When ${\rm Prim}(\A)$ is a (locally compact) Hausdorff space (and in this case ${\rm Prim}(\A)\simeq\hat{\A}$) one can \emph{localize} at a given point $\rr{I}\in {\rm Prim}(\A)$ considering a compact neighborhood $\bb{F}$ of 
$\rr{I}$ and considering the quotient $\A^\bb{F}:=\A/\A_{{\rm Prim}(\A)\setminus \bb{F}}$. In fact, by the Hausdorff condition ${\rm Prim}(\A)\setminus \bb{F}$ is open in ${\rm Prim}(\A)$. Condition (ii) in Definition \ref{def:CTCA} means that for every $\rr{I}\in {\rm Prim}(\A)$ there is a 
compact neighborhood $\bb{F}$ of 
$\rr{I}$ such that 
\[
\A^\bb{F}\;\simeq\; C(\bb{F})\otimes\rr{K}\;\simeq\;C(\bb{F},\rr{K})\;.
\]
We need a final observation. For a separable $C^*$-algebra  its spectrum is a second-countable space \cite[Proposition 3.3.4]{dixmier-77}. 
Then for a continuous trace $C^*$-algebra $\A$ its spectrum $\hat\A$ is  Hausdorff and second-countable. If $\bb{F}\subset \hat\A$ is a compact subset then $\bb{F}$ is compact, Hausdorff and  second-countable, hence metrizable (Urysohn metrization Theorem).
Then,  separable continuous-trace $C^*$-algebra are at least locally of the form of the 
  $C^*$-algebras considered in Sections \ref{sec:stanFer} and \ref{asec:class_eig} and the $C^*$-algebras of the type \eqref{eq:nic-alg} are special examples of (separable) continuous-trace $C^*$-algebras \cite[Example 5.18]{Rae}. Under these conditions Theorem \ref{prop: eigenstate fermi_S} applies at least locally and one would expect to cook up a global version by some technical gluing procedure.

%---Appwndix-----&
\appendix

\section{Disintegration of measures}
\label{sec:disint}

In this appendix, we briefly review the theory of disintegration of measures between compact metric (or metrizable) spaces. The results of this section are adapted from %\cite[Section 452]{Frem} and 
\cite[Section VII.2]{conway2}.

\medskip

 Let $(X,\bb{B}_X, \mu)$ be a Borel measure space, 
 with $X$
  a compact metric space\footnote{With a little more generality we can require   that $X$ is  a Radon space.}  endowed with
 its Borel $\sigma$-algebra $\bb{B}_X$ and $\mu:\bb{B}_X\to [0,+\infty]$ a positive regular measure. We will denote with $\rr{M}(X)$ the space of  regular Borel  measures on $X$.
  We say that $(X,\bb{B}_X, \mu)$ is a   probability space when $\mu(X)=1$.    Let   $(Y,\bb{B}_Y)$ be a second compact metric  space endowed with its Borel $\sigma$-algebra and 
  $f:X\to Y$ a Borel function. The \emph{pushforward} of the measure $\mu$ by $f$ is the measure 
  $\nu:\bb{B}_Y\to [0,+\infty]$ defined by $\nu(\Sigma):=\mu(f^{-1}(\Sigma))$ for every Borel set $\Sigma\in \bb{B}_Y$. In the following, we will use the short notation $\nu=\mu\circ f^{-1}$. 
Let us denote with $\bb{L}^p(X,\mu)$  the space of $p$-integrable functions and by   $L^p(X,\mu)$ the corresponding Lebesgue space of equivalence classes of these functions.
 For each $\psi\in \bb{L}^1(X,\mu)$ the map 
 \[
 L^\infty(Y,\nu)\;\ni\; \phi\;\longmapsto\; \int_X{\rm d}\mu(x)\;\psi(x)\;(\phi\circ f)(x)\;\in\;\C
 \]
 defines a bounded linear functional on $L^\infty(Y,\nu)$.  Let $\chi_\Sigma\in L^\infty(Y,\nu)$ be the characteristic function of $\Sigma$. Then the mapping
 \[
\bb{B}_Y\;\ni\; \Sigma\;\longmapsto\;\int_X{\rm d}\mu(x)\;\psi(x)\;(\chi_\Sigma\circ f)(x)\;=\;\int_{f^{-1}(\Sigma)}{\rm d}\mu(x)\;\psi(x)\;\in\;\C
 \]
  is a countably additive (complex-valued) measure  on $Y$ that is absolutely continuous with respect to $\nu$. Therefore,  there is a unique element $E_\psi\in L^1(Y,\nu)$, the so-called  Radon–Nikodym derivative, such that
 \begin{equation}\label{eq:disin_01}
 \int_X{\rm d}\mu(x)\;\psi(x)\;(\phi\circ f)(x)\;=\;\int_Y{\rm d}\nu(y)\; E_\psi(y)\;\phi(y) \;,\qquad\forall\;\phi\in L^\infty(Y,\nu)\;.
\end{equation}
 This construction provides a well defined  map  $E\colon {L}^1(X,\mu)\to L^1(Y,\nu)$  called the \emph{expectation operator}. As proved in  \cite[Section VII.2, Proposition 2.8]{conway2} this is a contraction, \ie $\|E_\psi\|_{L^1}\leqslant\|\psi\|_{L^1}$ for every $\psi\in {L}^1(X,\mu)$.

\begin{definition}
 A disintegration of the measure $\mu$ with respect to $f:X\to Y$ is a function
 $Y\ni y\mapsto \mu_y \in\rr{M}(X)$ such that: (i) for each $y\in Y$, $\mu_y$ is a positive probability measure on $X$; (ii) if $\psi\in  {L}^1(X,\mu)$ then
 \begin{equation}\label{eq:dis_dis}
    E_\psi(y)\;=\;\int_X{\rm d}\mu_y(x)\;\psi(x)\;
\end{equation}
where the equality is meant $\nu$-almost everywhere. 
\end{definition}

  \medskip
    
 Let    $\{\mu_y\}_{y\in Y}$ be a disintegration $\mu$  with respect to $f:X\to Y$ and $\Delta\in \bb{B}_X$ with characteristic function $\chi_\Delta$.
 Starting from \eqref{eq:disin_01} and setting $\phi\equiv 1$, a direct computation   shows that 
\[
    \begin{split}
        \mu(\Delta)\;&=\;\int_X{\rm d}\mu(x)\;\chi_\Delta(x)\;=\;\int_Y{\rm d}\nu(y)\;E_{\chi_\Delta}(y)\\
        &=\;\int_Y{\rm d}\nu(y)\left(\int_X{\rm d}\mu_y(x)\;\chi_\Delta(x)\right)\;=\;\int_Y\,{\rm d}\nu(y)\;\mu_y(\Delta)\;.
    \end{split}
\]   
So the disintegration $\{\mu_y\}_{y\in Y}$ does indeed \emph{disintegrate} the measure $\mu$ into the pieces $\mu_y$.
Moreover, for every $\psi\in  {L}^1(X,\mu)$ a similar computation provides
\begin{equation}\label{eq:disin_02}
    \begin{split}
     \int_X{\rm d}\mu(x)\;\psi(x)\;&=\;\int_Y{\rm d}\nu(y)\;E_{\psi}(y)\;=\;\int_Y{\rm d}\nu(y)\left(\int_X{\rm d}\mu_y(x)\;\psi(x)\right)\;.
    \end{split}
\end{equation}
From its very definition it also follows that each measure    $\mu_y$  of a disintegration is carried by the level set $f^{-1}(\{y\})\subseteq X$ for
    $\nu$-almost every $y\in Y$ \cite[Section VII.2, Proposition 2.8]{conway2}. More precisely one has that 
\[
\mu_y\left(X\setminus f^{-1}(\{y\})\right)\;=\;0
\]
   for
    $\nu$-almost every $y\in Y$.
    This implies that $\mu_y(\Delta)=\mu_y(\Delta\cap f^{-1}(\{y\}))$
for every     $\Delta\in \bb{B}_X$. Combining the latter fact with \eqref{eq:disin_02} one gets
\begin{equation}\label{eq:disin_03}
    \begin{split}
     \int_X{\rm d}\mu(x)\;\psi(x)\;&=\;\int_Y\int_{f^{-1}(y)}{\rm d}\nu(y)\;{\rm d}\mu_y(x)\;\psi(x)
    \end{split}
\end{equation}
    for every $\psi\in  {L}^1(X,\mu)$, or more in general for every Borel-measurable function such that $|\psi|:X\to[0,+\infty]$.

      \medskip

The next result guarantees the existence and the uniqueness of the disintegration of a measure \cite[Section VII.2, Proposition 2.8]{conway2}. 

\begin{theorem}\label{teo: desintegration}(Disintegration Theorem)
Let $(X,\bb{B}_X, \mu)$ be a Borel measure space, given by  a compact metric space $X$ endowed with
 its Borel $\sigma$-algebra $\bb{B}_X$ and a positive regular probability measure
 $\mu:\bb{B}_X\to [0,+\infty]$.   Let   $(Y,\bb{B}_Y)$ be a second compact metric  space endowed with its Borel $\sigma$-algebra and 
  $f:X\to Y$ a Borel function. Then, there exists a disintegration $\{\mu_y\}_{y\in Y}$ of
   $\mu$  with respect to $f:X\to Y$. Moreover, if there is a second 
disintegration $\{\mu_y'\}_{y\in Y}$ of
   $\mu$  with respect to the same $f$ then $\mu_y=\mu_y'$ for
    $\nu$-almost every $y\in Y$.
\end{theorem}

\medskip

It is useful in applications to know   the behavior of the disintegration under the composition of measurable maps.

\begin{proposition}\label{prop: composition}
    Let $(X_i,\mu^i)$ be Borel probability spaces for $i=1,2,3$. Let $f\colon X_1\to X_2$ and $g\colon X_2\to X_3$ be two measurable maps so that $\mu^2$ and $\mu^3$ are the pushforward measures induced by $f$ and $g$, respectively. If $h:=g\circ f$, then the expectation operator $E^h$ with respect to $h$ fulfills the composition rule $E^h=E^gE^f$, where $E^g$ and $E^f$ are the expectation operators associated with $f$ and $g$, respectively. Furthermore, the disintegration $\{\mu^{1,h}_z\}_{z\in X_3}$ of $\mu^1$ with respect to $\mu^3$ related to $h$ meets the equality
    \begin{equation*}
        \mu_z^{1,h}(\Sigma)\;=\;\int_{X_2}\dd\mu^{2,g}_{z}(y)\;\mu^{1,f}_y(\Sigma)\;,
    \end{equation*}
  for $\mu^3$-almost all $z\in X_3$, and any measurable sets $\Sigma\subset X_1$.  Here $\{\mu_z^{2,g}\}_{z\in X_3}$ and $\{\mu_y^{1,f}\}_{y\in X_2}$
  denote the  disintegrations of $\mu^2$ and $\mu^1$
  related to $g$ and  $f$, respectively.
\end{proposition}
\begin{proof}First of all, observe that $\mu^3$ coincides with the pushforward measure of $\mu^1$ induced by $h.$ Thus, one has that
    \begin{equation}\label{equ: 1}
        \int_{X_1}\dd\mu^1(x)\;\phi(h(x))\psi(x)\;=\;\int_{X_3}\dd\mu^3(z)\;\phi(z)(E^h_\psi)(z)
    \end{equation}
for all $\phi\in L^\infty(X_3,\mu^3)$ and $\psi\in L^1(X_1,\mu^1).$    On the other hand, let us note that the left-hand side of \ref{equ: 1} satisfies
    \begin{equation}\label{equ: 2}
    \begin{split}
           \int_{X_1}\dd\mu^1(x)\;\phi(h(x))\psi(x)\;&= \;   \int_{X_1}\dd\mu^1(x)\;(\phi\circ g)(f(x))\psi(x)\\
           &=\;\int_{X_2}\dd\mu^2(y)\;(\phi \circ g)(y)(E^f_\psi)(y)\\
             &=\;\int_{X_2}\dd\mu^2(y)\;\phi (g(y))(E^f_\psi)(y)\\
             &=\;\int_{X_3}\dd\mu^3(z)\;\phi(z) (E^g_{E^f_\psi})(z)\;.
    \end{split}
    \end{equation}
With the identification $E^g_{E^f_\psi}:=(E^gE^f)_\psi$, and in light of \ref{equ: 1} and \ref{equ: 2}, it holds that $E^h=E^gE^f$ 
as composition of linear maps between the related $L^1(X^i,\mu^i)$
spaces.
Moreover,
 \begin{equation*}
     \begin{split}
      \mu^{1,h}_{z}(\Sigma)\;&=\;(E^h_{\chi_\Sigma})(z)\;=\;\left((E^gE^f)_{\chi_\Sigma}\right)(z)\;=\;\int_{X_2}\dd\mu_z^{2,g}(y)\;(E^f_{\chi_\Sigma})(y)\\
      &=\;\int_{X_2}\dd\mu_z^{2,g}(y)\;\mu_y^{1,f}(\Sigma)
     \end{split}
 \end{equation*}
 for every measurable set $\Sigma\subseteq X^1$.
\end{proof}

\medskip

Let us end this section with some simple but useful examples.
\begin{example}[Disintegration along a product space]
  Let $(X_1,\mu_1)$ and $(X_2,\mu_2)$ be two compact metric spaces endowed with regular probability measures. Let us consider the product space $(X\times Y,\mu)$ with  the corresponding product measure $\mu:=\mu_1\times \mu_2$ and the     projection on the first component $\pi: X_1\times X_2\to X_1$ defined by $\pi(x_1,x_2)=x_1$ for every $x_j\in X_j$ and $j=1,2$. Let $\nu$ be the pushforward  of the measure $\mu$ by $\pi$. From its very definition it follows that for every 
   measurable set $\Sigma\subseteq X_1$ one gets
    \[
    \nu(\Sigma)\;=\;\mu\left(\pi^{-1}(\Sigma)\right)\;=\;\mu(\Sigma\times X_2)\;=\;\mu_1(\Sigma)\mu_2(X_2)\;=\;\mu_1(\Sigma)\;.
    \]
    Therefore one has that $\nu=\mu_1$.
The next task is to calculate the expectation operator. For that, consider a $\psi\in L^1(X_1\times X_2,\mu)$. Then, in view of \eqref{eq:disin_01} and for all $\phi\in L^\infty(X_1,\mu_1)$ it holds that
    \begin{equation*}
        \begin{split}
            \int_{X_1}{\rm d}\mu_1(x_1)\; E_\psi(x_1)\phi(x_1)\;&=\;\int_{X_1\times X_2}{\rm d}\mu(x_1,x_2)\; \psi(x_1,x_2)\phi(\pi(x_1,x_2))\\
            &=\;\int_{X_1\times X_2}{\rm d}\mu(x_1,x_2)\;\psi(x_1,x_2)\phi(x_1)\\
            &=\;\int_{X_1}{\rm d}\mu_1(x_1)\; \phi(x_1)\left(\int_{X_2}{\rm d}\mu_2(x_2)\psi(x_1,x_2)\right)\;,
        \end{split}
    \end{equation*}
where the last  line is a consequence of Fubini's theorem. Since
\[
 \int_{X_1}{\rm d}\mu_1(x_1)\; \phi(x_1)\left[E_\psi(x_1)-\int_{X_2}{\rm d}\mu_2(x_2)\psi(x_1,x_2)\right]\;=\;0
\]
for every  $\phi\in L^\infty(X_1,\mu_1)$, it follows that 
\[
E_\psi(x_1)\;=\;\int_{X_2}{\rm d}\mu_2(x_2)\psi(x_1,x_2)
\]
as elements of $L^1(X_2,\mu_2)$.
To compute the disintegration $\{\mu_{x_1}\}_{x_1\in X_1}$ of $\mu$ let us use \eqref{eq:dis_dis}.
Let $\Sigma\subseteq X_1$ and $\Gamma\subseteq X_2$ be two measurable sets and $\chi_{\Sigma\times \Gamma}$  the characteristic function of the product $\Sigma\times \Gamma$.
Then, one has that
\begin{equation}
    \mu_{x_1}(\Sigma\times \Gamma)\;=\;E_{\chi_{\Sigma\times \Gamma}}(x_1)\;=\;\int_{X_2}{\rm d}\mu_2(x_2)\,\chi_{\Sigma\times \Gamma}(x_1,x_2)\;=\;\mu_2(\Gamma)\chi_\Sigma(x_1)\;,
\end{equation}
where in the last equality it has been used that $\chi_{\Sigma\times \Gamma}=\chi_{\Sigma}\chi_{\Gamma}$. Therefore, one gets that
\[
\mu_{x_1}\;=\;\delta_{x_1}\times \mu_2
\]
is a product measure where 
$\delta_{x_1}$ denotes the  the Dirac measure on $X_1$ concentrated on the point $x_1$.
\hfill $\blacktriangleleft$
\end{example}
 \begin{example}[Disintegration on the real line]\label{dis: function}
Let $X:=[a,b]\subset\R$ endowed with the normalized Lebesgue measure $\mu$, and $f:X\to\R$ a $C^1$-function. Let $a<x_1<\ldots<x_{N}<b$ be the set of points in which $f'(x_k)=0$
(it is not necessary to assume that $f'$ vanishes on the extremes $a$ and $b$). The set of intervals $I_k:=[x_{k},x_{k+1}]$ with $k=0,\ldots,N$, $x_0:=a$ and $x_{N+1}:=b$, provides a partition of $X$ such that $f$ is strictly monotone inside each $I_k$. 
The set $Y:=f(X)=[f_{\rm min},f_{\rm max}]\subset\R$ is again an interval delimited by the maximum and minimum values of $f$ over $X$. In particular the values $f_{\rm min}$ and $f_{\rm max}$  
are attained by evaluating $f$ over some of the points $x_k$.
It is also true that for every $y\in Y$ the preimmage $P_y:=f^{-1}(\{y\})$ is a finite, hence discrete, set of $X$.
 Let $\nu$ be the pushforward  of the measure $\mu$ by $f$. From its very definition it follows that for every 
   measurable set $\Sigma\subseteq Y$ one gets
    \[
    \begin{aligned}
    \nu(\Sigma)\;&=\;\mu\left(f^{-1}(\Sigma)\right)\;=\;\frac{1}{b-a}\int_a^b\dd x\;\chi_{f^{-1}(\Sigma)}(x)\\
    &=\frac{1}{b-a}\sum_{k=0}^N\int_{I_k}\dd x\;\chi_{f^{-1}(\Sigma)}(x)
    \end{aligned}
    \]
        Observe that  $f$ is invertible in each $I_{k}$. Therefore, by setting $f_k:=f|_{I_k}$ one can use  the change of variable $y=f(x)$ in every integral by obtaining
    \[
    \begin{aligned}
    \int_{I_k}\dd x\;\chi_{f^{-1}(\Sigma)}(x)\;&=\;\int_{f(x_k)}^{f(x_{k+1})}{\dd y}\;{\big(f_k^{-1}\big)'(y)}\;\chi_{f^{-1}(\Sigma)}\big(f^{-1}_k(y)\big)\\
    &=\;\int_{Y}{\dd y}\;\frac{\chi_{f(I_k)}(y)}{\left|f'\big(f^{-1}_k(y)\big)\right|}\;\chi_{\Sigma}(y)\;.
    \end{aligned}
    \]
 Therefore, one obtains that $\nu$ is absolutely continuous with respect to the Lebesgue measure on $Y$ and can be written in the form $\dd\nu(y)=\rho(y)\dd y$ with
 \[
 \rho(y)\;:=\;\frac{1}{(b-a)}\sum_{k=0}^N \frac{\chi_{f(I_k)}(y)}{\left|f'\big(f^{-1}_k(y)\big)\right|}\;.
 \]
 The next task is to calculate the expectation operator.  For that, consider a $\psi\in L^1(X,\mu)$ and a $\phi\in L^\infty(Y,\nu)$. Then, in view of \eqref{eq:disin_01} one gets
\[
\begin{aligned}
\int_Y{\rm d}y\;\rho(y)\; E_\psi(y)\;\phi(y)\;&=\;\int_X\dd\mu(x)\;\psi(x)\;(\phi\circ f)(x)\;.
\end{aligned}
\]
Bay decomposing the integral on the right-hand side over the intervals $I_k$, and applying again the change of variables $y=f(x)$, one gets
\[
\begin{aligned}
\int_Y{\rm d}y\;\rho(y)\; E_\psi(y)\;\phi(y)\;&=\;\int_Y\dd y\;\left(\frac{1}{(b-a)}\sum_{k=0}^N \frac{\chi_{f(I_k)}(y)\psi\big(f^{-1}_k(y)\big)}{\left|f'\big(f^{-1}_k(y)\big)\right|}\right)\;\phi(y)\;.
\end{aligned}
\]
Since the equality  above must hold independently of $\phi$ it turns out that 
\[
\sum_{k=0}^N \frac{\chi_{f(I_k)}(y)\big[E_\psi(y)-\psi\big(f^{-1}_k(y)\big)\big]}{\left|f'\big(f^{-1}_k(y)\big)\right|}\;=\;0
\] 
for almost every $y$ in $Y$. Therefore, one has that
\[
E_\psi(y)\;=\;\sum_{x\in P_y}c_x(f)
\psi(x)\;,\qquad c_k(f)\;:=\;\left(\sum_{x\in P_y} \frac{1}{\left|f'(x)\right|}\right)^{-1} \frac{1}{\left|f'(x)\right|}\;.
\]
Finally, the disintegration $\{\mu_{y}\}_{y\in Y}$ of $\mu$ is obtained  by using  \eqref{eq:dis_dis} which provides
\[
\int_X{\rm d}\mu_y(x)\;\psi(x)\;=\;\sum_{x\in P_y}c_x
\psi(x)
\]
showing that
\[
\mu_y\;=\;\sum_{x\in P_y}c_x\delta_{x}
\]
 coincides with a convex combination of Dirac measures supported on the points of $P_y$.
 As  a special application of this result let us consider the 
 trigonometric expression 
 \[
f(x)\;:=\;\sqrt{\gamma^2+4\cos^2\left(\frac{x}{2}\right)}\;,\qquad x\in[-2\pi,2\pi]
\]
  with $\gamma>0$. In this case, the critical points are $x_k:=(k-2)\pi$ for $n=0,1,2,3,4$ and one can define the associated partition $I_k$.
 For every $y\in Y:=(\gamma,\sqrt{\gamma^2+4})$ the set of points $P_y:=f^{-1}(\{y\})$ is given by four points of the type $\{\pm x_\ast,\pm2\pi\mp x_\ast\}$ where the reference point $x_\ast=x_\ast(y)$ is chosen as $x_\ast:=f^{-1}(\{y\})\cap[0, {\pi}]$. 
   To compute the coefficients $c_x(f)$, let us observe that 
   \[
   f'(x)\;:=\;-\frac{\sin(x)}{f(x)}\;.
   \]
  By observing that $| f'(x)|=|\sin(x_*)|f(x_*)^{-1}$ for every $x\in P_y$ it follows that $c_x(f)=\frac{1}{4}$. As a result one gets
 \[
\mu_y\;=\;\frac{1}{4}\sum_{\sharp=\pm}\delta_{\sharp x_\ast\;+\;\delta_{\sharp(2\pi-x_\ast)}}
\]
  for the associated disintegration of the normalized Lebesgue measure $\mu$ on $[-2\pi,2\pi]$.
    \hfill $\blacktriangleleft$
\end{example}

\begin{example}[Disintegration by a linear function]\label{ex: a}
Given  $\ell>0$, let us consider the square $X:=[0,\ell]\times[0,\ell]\subset\R^2$ endowed with its normalized (product) Lebesgue measure
$\dd\mu:=\ell^{-2}\dd x_1 \dd x_2$. Consider the linear function
on $X$ given by $f(x_1,x_2):=x_2-x_1$. It is not hard to see that
 the range $f(X)$ coincides with the interval $Y:=[-\ell ,\ell]$. Therefore, one can see $f$ as a map $f:X\to Y$. Since $f$ is continuous, hence measurable,  we can compute the 
  pushforward  measure $\nu$ on $Y$ induced by the measure $\mu$ through $f$.  Let $\Sigma\subseteq Y$ by any measurable set and $\chi_\Sigma$ its characteristic function. Then, by definition
  \[
  \begin{aligned}
  \nu(\Sigma)\;:&=\;\mu\left(f^{-1}(\Sigma)\right)\;=\;\int_{X}\dd\mu(x_1,x_2)\chi_{f^{-1}(\Sigma)}(x_1,x_2)\\
  &=\;\frac{1}{\ell^2}\int_{0}^{\ell}\dd x_1\int_{0}^{\ell}\dd x_2\;\chi_\Sigma(x_2-x_1)
   \end{aligned}
  \]
  where the last equality is a consequence of Fubini's theorem.
Let us assume for the moment that $\chi_\Sigma$ is Riemann integrable and consider the change of variables $(s,y)\mapsto(x_1,x_2)$ given by  
 $x_1(s,y):=s$ and  $x_2(s,y):=s+y$. The Jacobian determinant of the transformation is $1$, and after some manipulation one gets 
 \begin{equation}\label{eq:indu:mwes}
 \begin{aligned}
  \nu(\Sigma)\;&=\;
 \frac{1}{\ell^2}\int_{0}^{\ell}\dd s\left(\int_{-s }^{\ell-s}\dd y\;\chi_\Sigma(y)
\right) \\
&=\;
 \frac{1}{\ell}\int_{0}^{\ell}\dd s\;g_\Sigma(s)
 \end{aligned}
\end{equation}
  where
 \[
 g_\Sigma(s)\;:=\;\frac{1}{\ell}\int_{-s }^{\ell-s}\dd y\;\chi_\Sigma(y)\;=\;\frac{\big|[-s,\ell-s]\cap\Sigma\big|}{\ell}
 \]
 and $|\cdot|$ in the numerator denotes the usual Lebesgue measure on $\R$. Evidently, it holds true that $g_\Sigma:[0,\ell]\to[0,1]$. 
The formula above and the Lebesgue’s criterion for integrability imply that  $g_\Sigma=0$ whenever  $\Sigma$ is any measurable set of zero Lebesgue measure. Therefore  from \eqref{eq:indu:mwes} one infers that $\nu$  is  absolutely continuous with respect to the Lebesgue measure $\dd y$ on $Y$. Therefore, in view of the Radon–Nikodym theorem, there exists a measurable function $\rho:Y\to [0,\infty)$ such that $\dd\nu(y)=\rho(y)\dd y$.
Let us observe that the integral \eqref{eq:indu:mwes} can be rewritten as
\[
 \begin{aligned}
 \nu(\Sigma)\;&=\; \frac{1}{\ell^2}\int_{0}^{\ell}\dd s\left(\int_{-\ell }^{\ell}\dd y\;\chi_{[-s,\ell-s]}(y)\chi_\Sigma(y)\right)\\
 &=\; \frac{1}{\ell^2}\int_{Y}\dd y\; \chi_\Sigma(y)\left(\int_{0}^{\ell}\dd s\;\chi_{[-s,\ell-s]}(y)
\right)
 \end{aligned}
\]
where in the second equality the order of integration has been exchanged.
Therefore one gets that
\[
\rho(y)\;:=\;\frac{1}{\ell^2}\int_{0}^{\ell}\dd s\;\chi_{[-s,\ell-s]}(y)\;=\;\frac{\ell-|y|}{\ell^2}
\]
which provides the expression for the Radon–Nikodym derivative.
The next task is to calculate the expectation operator. For that, consider a $\psi\in L^1(X,\mu)$. Then, in view of \eqref{eq:disin_01} and for all $\phi\in L^\infty(Y,\nu)$ it holds that
    \begin{equation*}
        \begin{split}
            \int_{Y}{\rm d}y\; \rho(y)\phi(y)E_\psi(y)\;&=\;\frac{1}{\ell^2}\int_{0}^\ell\int_{0}^\ell{\rm d}x_1{\rm d}x_2\; \psi(x_1,x_2)\phi(x_2-x_1)\;.
        \end{split}
    \end{equation*} 
To massage the right-hand side of this equation let us assume for the moment that the function $\psi$ and $\phi$ are Riemann integrable so that change of variable  $(s,y)\mapsto(x_1,x_2)$ can be used again. Then
 \[
 \begin{aligned}
 \int_{Y}{\rm d}y\; \rho(y)\phi(y)E_\psi(y)\;&=\;\frac{1}{\ell^2}\int_{0}^\ell{\rm d}s\int_{-s}^{\ell-s}{\rm d}y\; \psi(s,s+y)\phi(y)\\
 &=\;\frac{1}{\ell^2}\int_{0}^\ell{\rm d}s\int_{Y}{\rm d}y\; \chi_{[-s,\ell-s]}(y)\psi(s,s+y)\phi(y)\\
 &=\;\int_{Y}{\rm d}y\; \phi(y)\left(\frac{1}{\ell^2}\int_{0}^\ell{\rm d}s\; \chi_{[-s,\ell-s]}(y)\psi(s,s+y)\right)\;.
 \end{aligned}
 \]
 The equality between the first and the last member holds for very bounded Riemann measurable function $\phi$, and in view of the Lebesgue’s criterion for every $\phi\in L^\infty(Y,\nu)$. This implies that
 \[
 \rho(y)E_\psi(y)\;=\;\frac{1}{\ell^2}\int_{0}^\ell{\rm d}s\; \chi_{[-s,\ell-s]}(y)\psi(s,s+y)
 \]
for almost all $y$ in $Y$. With a final manipulation one gets
 \[
 E_\psi(y)\;=\;\frac{1}{\ell-|y|}
 \left\{
 \begin{aligned}
 &\int_{0}^{\ell-y}\dd s\;\psi(s,s+y)&\;\;& &0\leqslant y\leqslant \ell&\\
  &\int_{|y|}^{\ell}\dd s\;\psi(s,s+y)&\;\;& &-\ell\leqslant y\leqslant 0&\;.\\
 \end{aligned}
 \right.
 \]
 Let us compare the last equation with the \eqref{eq:dis_dis}, here rewritten as
 \[
 E_\psi(y)\;=\;\int_{X}\dd\mu_y(x_1,x_2)\;\psi(x_1,x_2)\;,
 \]
where $\{\mu_y\}_{y\in Y}$ is the disintegration of $\mu$  with respect to $f:X\to Y$. Then one gets that $\mu_y$ is concentrated on the segment 
\[
L_y\;:=\;f^{-1}(\{y\})\;=\;\{(x_1,x_2)\in\R^2\;|\; x_2-x_1= y\}\;\cap\; [0,\ell]^2
\]
and over $L_y$ it agrees with the normalized one-dimensional Lebesgue measure. Let $[0,1]\ni\tau \mapsto s(\tau)\in L_y$ be the linear reparametrization of $L_y$, given by $s(\tau)=(\alpha +\beta\tau,\alpha +\beta\tau+y) $
and $\alpha,\beta$ depending on $y$. Then, with an innocent misuse of notation, and by using the Dirac measure $\R^2$, one can write
\[
\mu_{y}(x_1,x_2)\;=\;\int_0^1\dd \tau\; \delta_{s(\tau)}(x_1x_2)
\]
for the disintegration of the measure $\mu$ along $f$.
\hfill $\blacktriangleleft$
\end{example}

%------------------------------------------------------------------------------------------------------------------------%
%                                                                        bibliography
%------------------------------------------------------------------------------------------------------------------------%

%------------------------------------------------------------------------------------------------------------------------%
\end{document}